\documentclass[prodmode,acmtocl]{acmsmall}
\pdfoutput=1  
\usepackage[colorlinks,urlcolor=black]{hyperref}
\usepackage{mdwlist} 

\usepackage{amssymb,stmaryrd,amsmath,my-amsthm,xspace}
\usepackage{nccmath}
\usepackage{prooftree}
\usepackage{graphicx}
\usepackage[colorlinks,urlcolor=black]{hyperref}
\usepackage{soul}

\usepackage{microtype}
\usepackage{tgpagella}

\usepackage{array}
\newcolumntype{L}[1]{>{$}p{#1}<{$}}
\newcolumntype{C}[1]{>{\centering$}p{#1}<{$}}
\newcolumntype{R}[1]{>{\raggedleft$}p{#1}<{$}}
\newcommand\maketab[2]{\newenvironment{#1}{\begin{quote}\noindent\begin{tabular}{#2}}{\end{tabular}\end{quote}}}

\usepackage{graphics} 
\usepackage{type1cm}
\usepackage{eso-pic}

\usepackage{float}
\floatstyle{boxed} \restylefloat{figure}

\usepackage{fancybox}
\newlength{\mylength}
\newenvironment{frameqn}%
{\setlength{\fboxsep}{4pt}
\setlength{\mylength}{\linewidth}%
\addtolength{\mylength}{-2\fboxsep}%
\addtolength{\mylength}{-2\fboxrule}%
\Sbox
\minipage{\mylength}%
$$}%
{$$\endminipage\endSbox
{
\[\fbox{\TheSbox}\]
}}

\newenvironment{frametxt}%
{
\setlength{\fboxsep}{4pt}
\setlength{\mylength}{\linewidth}%
\addtolength{\mylength}{-2\fboxsep}%
\addtolength{\mylength}{-2\fboxrule}%
\Sbox
\minipage{\mylength}%
}%
{\endminipage\endSbox
{
\[\fbox{\TheSbox}\]
}}

\newcommand\shs{\mathrel{\Delta}}
\def\equiv{=}
\newcommand\raws{{\mathtt{s}}}
\newcommand\rawr{{\mathtt{r}}}
\newcommand\rawphi{{\mathtt{phi}}}
\newcommand\rawpsi{{\mathtt{psi}}}
\newcommand\dottimes{\mathrel{\text{$\dot*$}}}
\newcommand\dotzero{\text{$\dot 0$}}
\newcommand\dotsucc{\text{$\dot{\tf{succ}}$}}
\newcommand\dotbot{\text{$\dot\bot$}}
\newcommand\dotforall{\text{$\dot\forall$}}
\newcommand\dotlimp{\mathrel{\text{$\dot\limp$}}}
\newcommand\dotoeq{\mathrel{\text{$\dot\oeq$}}}

\newcommand{\denot}[3]{\llbracket #3 \rrbracket_{\scalebox{.6}{$#2$}}^{\hspace{-.1ex}\scalebox{.4}{$#1$}}}

\newcommand\den[1]{{\hspace{-.1ex}\scalebox{.45}{$#1$}}}
\newcommand\iden{{\scalebox{.4}{$\mathcal I$}}}
\newcommand\nden{\den{\mathcal N}}
\newcommand\mden{\den{\mathcal M}}

\newcommand\Deriv{\Delta}
\newcommand\pmss{\f{pmss}}
\newcommand\sort{\f{sort}}
\newcommand\somerel{\mathrel{\mathcal R}}
\newcommand\theory[1]{\ensuremath{\mathsf{#1}}}
\newcommand\pnleq{\approx}
\newcommand\oeq{\approx}

\newcommand\Chi{\raisebox{.15em}{\large$\chi$}}    
\newbox\tempa
\newbox\tempb
\newdimen\tempc
\def\mud#1{\hfil $\displaystyle{\mathstrut #1}$\hfil}
\def\rig#1{\hfil $\displaystyle{#1}$}
\def\irulehelp#1#2#3{\setbox\tempa=\hbox{$\displaystyle{\mathstrut #2}$}%
		        \setbox\tempb=\vbox{\halign{##\cr
	\mud{#1}\cr
	\noalign{\vskip\the\lineskip}%
	\noalign{\hrule height 0pt}%
	\rig{\vbox to 0pt{\vss\hbox to 0pt{${\; #3}$\hss}\vss}}\cr
	\noalign{\hrule}%
	\noalign{\vskip\the\lineskip}%
	\mud{\copy\tempa}\cr}}%
		      \tempc=\wd\tempb
		      \advance\tempc by \wd\tempa
		      \divide\tempc by 2 }
\def\irule#1#2#3{{\irulehelp{#1}{#2}{#3}%
		     \hbox to \wd\tempa{\hss \box\tempb \hss}}}

\newcommand{\fa}{\forall}

\newcommand\basesort{\tau}
\newcommand\ns[1]{{\mathsf{#1}}}

\newcommand{\model}[1]{\denot{\mathcal I}{}{#1}}
\newcommand{\amod}[1]{\ensuremath{(#1)^{.}}}

\newcommand\deffont[1]{{\bf #1}}
\newcommand\tf[1]{{\mathsf{#1}}}
\newcommand\f[1]{\mathit{#1}}

\newcommand\act{{\cdot}}
\newcommand\liff{\mathrel{\Leftrightarrow}}
\newcommand\limp{\Rightarrow}
\newcommand\Forall[1]{\forall #1.}
\newcommand\Exists[1]{\exists #1.}
\newcommand\lam[1]{\lambda #1.}
\newcommand\aeq{\mathrel{=_{\alpha}}}

\newcommand\id{\f{id}}
\newcommand\Id{\f{Id}}
\newcommand\cent{\vdash}

\newcommand\sm{{\mapsto}}
\newcommand\ssm{{{:}{:}{=}}}
\newcommand\mone{{}^{\text{-}1}}
\newcommand\rulefont[1]{\ensuremath{\bf (#1)}\xspace}
\newcommand\new{\reflectbox{\ensuremath{\mathsf{N}}}}
\newcommand\New[1]{\new #1.}

\newcommand\atomsup{\mathbb A^{\hspace{-.25ex}\scalebox{.6}{$>$}}}
\newcommand\atomsdown{\mathbb A^{\hspace{-.25ex}\scalebox{.6}{$<$}}}
\newcommand\atoms{{\mathbb A}}

\newtheoremstyle{jamiestyle}
  {4pt}
  {0pt}
  {\it}
  {0pt}
  {\bf}
  {.}
  { }
  {}
\theoremstyle{jamiestyle}
\newtheorem{thrm}{Theorem}[section]
\newtheorem{prop}[thrm]{Proposition}
\newtheorem{lemm}[thrm]{Lemma}
\newtheorem{corr}[thrm]{Corollary}
\newtheoremstyle{jamienfstyle}
  {4pt}
  {0pt}
  {\normalfont}
  {0pt}
  {\bf}
  {.}
  { }
  {}
\theoremstyle{jamienfstyle}
\newtheorem{nttn}[thrm]{Notation}
\newtheorem{defn}[thrm]{Definition}
\newtheorem{xmpl}[thrm]{Example}
\newtheorem{rmrk}[thrm]{Remark}

\begin{document}

\markboth{Dowek, Gabbay}{Permissive-nominal logic}

\title{Permissive-Nominal Logic (journal version)}
\author{\href{http://www-roc.inria.fr/who/Gilles.Dowek/}{Gilles Dowek}\  
\affil{INRIA, France. \qquad\ \ \url{http://www-roc.inria.fr/who/Gilles.Dowek}}
\href{http://www.gabbay.org.uk}{Murdoch J. Gabbay}
\affil{Heriot-Watt University, UK. \qquad \url{http://www.gabbay.org.uk}}
}

\begin{abstract}
Permissive-Nominal Logic (PNL) is an extension of first-order predicate
logic in which term-formers can bind names in their arguments.

This allows for direct axiomatisations with binders, such as of the $\lambda$-binder of the lambda-calculus or the $\forall$-binder of first-order logic.  
It also allows us to finitely axiomatise arithmetic, and similarly to axiomatise `nominal' datatypes-with-binding. 

Just like first- and higher-order logic, equality reasoning is not necessary to $\alpha$-rename.

This gives PNL much of the expressive power of higher-order logic, but models and derivations of PNL are first-order in character, and the logic seems to strike a good balance between expressivity and simplicity. 
\end{abstract}




\acmformat{Dowek, Gilles, Gabbay, Murdoch J., 2012. Permissive-nominal logic (journal version)}

\begin{bottomstuff}
Thanks to the anonymous referees.
Supported by grant RYC-2006-002131 at the Polytechnic University of Madrid.
\end{bottomstuff}

\maketitle

\tableofcontents

\section{Introduction}

In the early 20th century, the expressivity of logics was considered 
\emph{in principle}. 
For example, 
first-order predicate logics with or without
term-formers are equally expressive, in principle.

In the early 21st century, more attention is paid to what we like to call `ergonomics'.
First-order predicate logic with term-formers is more ergonomic than first-order
predicate logic without term-formers; terms, propositions and
proofs are shorter and more natural in the former than in the
latter.

Another imperative is for a \emph{weak} logic---the fewer `bells and whistles' we have to worry about, the easier it will be to verify, implement, and prove its meta-theoretic properties. 
This can be in tension with being ergonomic, as the example of first-order predicate logic with or without term-formers illustrates.

Thus there enters a fruitful design tension: we aim for logics that are so ergonomic that they `just work', yet so weak and well-behaved that we can still prove good properties for them.

Now we come to the issue of \emph{binding}.
Binding is ubiquitous in logical specifications in mathematics---binding features in function definitions via $\lambda$-abstraction, 
and binders are also used to define sets in comprehension, and to define finite and infinite sums, 
integrals, derivations, quantifiers, and so on.
A logic for mathematics that provides support for this central and essential notion, 
is likely to be more ergonomic than a logic that does not.

First-order logic is weak, computationally tractable, and has good theoretical properties (unification of first-order terms is decidable; proof-search is simple and well-understood; models are simple).
However, first-order logic is unergonomic in the sense that it does not admit term-formers that can bind.
Thus it is hard to give direct, finite, first-order axiomatisations of set theory, arithmetic, higher-order logic, or the $\lambda$-calculus.

This is one reason that e.g. higher-order logic is often used as a specification language in theory (see \cite{farmer:sevvst} for an excellent exposition) and implementations (like Isabelle \cite{paulson:isanst})---higher-order logic has a binder ($\lambda$) built-in to terms, and so is more ergonomic.

However, higher-order logic is also stronger than first-order logic, less computationally tractable, and models tend to be more complex.
The jump from first- to higher-order logic is quite large.

This motivates the study of direct extensions of first-order predicate logic with term-formers that can bind.
The topic of this paper is the construction of one such extension, which we call \emph{permissive-nominal logic} (\deffont{PNL}). 
PNL has a clear first-order flavour, and it admits term-formers that bind.

\subsubsection*{Technical summary and overview}

The main technical contributions of this paper are: the definition of permissive-nominal logic, in particular how it handles freshness side-conditions and how this permits the addition of universal quantification to nominal terms; the Tarski-style models we then construct; and the finite yet fully first-order axiomatisations of substitution, first-order logic, and arithmetic which we then exhibit.

Soundness, completeness, and cut-elimination follow by fairly routine arguments, but we see this as a good sign: that the definition of permissive-nominal logic remains close to first-order logic, while allowing terms with binders. 

An overview of this paper is as follows:
\begin{itemize}
\item
We introduce the syntax and derivation rules of permissive-nominal logic (Section~\ref{sect.pnl}).

The impatient reader can jump directly to Figure~\ref{Seq} on page~\pageref{Seq} and see that these derivation rules are virtually indistinguishable from those of first-order logic; only an extra side-condition in \rulefont{\forall L} hints at any difference.\footnote{The language of terms is significantly different, though; see Definition~\ref{defn.syntax}.} 
\item
We prove soundness and completeness (Theorems~\ref{thrm.soundness} and~\ref{thrm.completeness}) with respect to a suitable notion of \emph{permissive-nominal set} (Definition~\ref{defn.nominal.set}).
\item
We axiomatise arithmetic in PNL and prove a correctness and conservative extension result Theorem~\ref{thrm.arithmetic}.
The axiomatisation of arithmetic is \emph{finite}; the induction schema normally given in first-order logic arithmetic is represented by a single axiom with a universal quantification $\forall X$ over a permissive-nominal terms unknown.
\item
We indicate how to axiomatise nominal inductive datatypes, the $\new$-quantifier, and semantic freshness (Section~\ref{sect.more.examples}).
\item
We prove cut-elimination (Theorem~\ref{thrm.cut}).
The proof is virtually identical to that of first-order logic.
\end{itemize}

\subsubsection*{Permissive-nominal logic, nominal logic, and nominal terms}

Permissive-nominal logic follows the \emph{nominal logic} of \cite{pitts:nomlfo-jv} in its name, which coined the term `nominal', but nominal logic from \cite{pitts:nomlfo-jv} is a first-order theory of (set of axioms for) the equivariant Fraenkel-Mostowski sets and associated constructions used in \cite{gabbay:newaas-jv}. 
The syntax, semantics, and derivability of PNL are new, as indeed is the application to arithmetic. 

The term-language of PNL goes back to the \emph{nominal terms} of \cite{gabbay:nomu-jv}.
It is permissive, which means that the freshness contexts from \cite{gabbay:nomu-jv} become a kind of sorting constraint called \emph{permission sets}.
For more discussion see \cite{gabbay:perntu-jv} which introduced permissive-nominal terms, and amongst other things gave correspondences with nominal terms and also higher-order patterns.

In this paper we extend nominal syntax further by introducing $\f{shift}$-permutations (Definition~\ref{defn.permutation}).
Also, unlike \cite{gabbay:nomu-jv,gabbay:perntu-jv} in PNL there is quantification over unknowns $\forall X$.

This journal paper follows a conference paper \cite{gabbay:pernl}.
With respect to that paper, we have made the following changes:
\begin{itemize*}
\item
The treatment of $\alpha$-equivalence has been streamlined, leading to simplified proofs.
Two structural rules \rulefont{\alpha_L} and \rulefont{\alpha_R} have been eliminated.
\item
The rule \rulefont{\new} from \cite{gabbay:pernl} is now part of the axiom rule, further simplifying the proof-theory. 
\item
The notion of permutation includes $\f{shift}$-permutations; these permutations `shift all atoms up by one'.
Some readers will see in this a de Bruijn-like `shift' function \cite{abadi:exps}.
This gives desirable extra power to $\forall$-quantification and, perhaps surprisingly, turns out to be compatible with nominal techniques' characteristic \emph{small support} property.
\item
We include proofs of completeness by a standard term-model construction, and a sketch proof of cut-elimination.
\end{itemize*}

\section{Permissive-Nominal Logic}
\label{sect.pnl}

\subsection{Syntax}

\begin{defn}
\label{defn.sort.sig}
A \deffont{sort-signature} is a pair $(\mathcal A,\mathcal B)$ of \deffont{name} and \deffont{base} sorts.
$\nu$ will range over name sorts; $\basesort$ will range over base sorts.
A \deffont{sort language} is then defined by
\begin{frameqn}
\alpha ::= \nu \mid \basesort \mid 
(\alpha,\ldots,\alpha) \mid [\nu]\alpha
.
\end{frameqn}
We admit the possibility of empty tuples, so that $()$ is a sort (the \emph{unit sort}).
\end{defn}

\begin{xmpl}
Examples of base sorts are: `$\lambda$-terms',\ `formulae',\ `$\pi$-calculus processes',\ and `program environments', `functions', `truth-values', `behaviours',\ and `valuations'.

Examples of name sorts are `variable symbols',\ `channel names',\ or `memory locations'.
\end{xmpl}

\begin{defn}
\label{defn.term.signature}
A \deffont{term-signature} over a sort-signature $(\mathcal A,\mathcal B)$ is a tuple $(\mathcal F,\mathcal P,\f{ar})$ where:
\begin{itemize*}
\item
$\mathcal F$ and $\mathcal P$ are disjoint sets of \deffont{term-} 
and \deffont{proposition-formers}.
\item $\f{ar}$ assigns to each ${\tf f\in\mathcal F}$ a
\deffont{term-former arity} $(\alpha)\tau$
and to each $\tf P\in\mathcal P$ a \deffont{proposition-former arity}
$\alpha$, where $\alpha$ and $\tau$ are in the sort-language
determined by $(\mathcal A,\mathcal B)$.

We will write $((\alpha_1,\ldots,\alpha_n))\tau$ just as $(\alpha_1,\ldots,\alpha_n)\tau$.
\end{itemize*}
\label{defn.signature}
A \deffont{signature} $\mathcal S$ is then a tuple $(\mathcal A,\mathcal B,\mathcal F,\mathcal P,\f{ar})$.
\end{defn}

\begin{nttn}
We write $\tf f:(\alpha)\tau$ for $\f{ar}(\tf f)=(\alpha)\tau$ and similarly we write $\tf P:\alpha$ for $\f{ar}(\tf P)=\alpha$. 
\end{nttn}

\begin{rmrk}
The reader familiar with higher-order logic can read $\f{ar}(f)=(\alpha)\tau$ as $f:\alpha\to\tau$ and no harm will come of it.
We do not do this because we are following a first-order logic notation---and because we want to avoid any possible confusion that $(\alpha\to\alpha)\to\alpha$ might be a sort.
It is not. 
\end{rmrk}
 
\begin{xmpl}
A signature for the $\lambda$-calculus would have one name-sort $\nu$, one base sort $\iota$, and term-formers $\tf{lam}:([\nu]\iota)\iota$, $\tf{app}:(\iota,\iota)\iota$, and $\tf{var}:(\nu)\iota$.

A proposition-former for nominal freshness $\#$ would have arity $(\nu,\iota)$, though the arity $([\nu]\iota)$ would also be possible (this would handle more of the properties of names at the level of the logic).
More on this in Subsection~\ref{subsection.freshness}.

Plenty more examples of PNL theories will follow.
\end{xmpl}

\begin{defn}
\label{defn.atoms}
For each name sort $\nu$ fix a disjoint countably infinite set of \deffont{atoms} $\atoms_\nu$ (\emph{level 1 names}).

For each $\nu$ also fix a bijective function $f_\nu$ from $\atoms_\nu$ to the integers $\mathbb Z=\{0,\text{-}1,1,\text{-}2,2,\ldots\}$ 
(that we can do this follows from our assumption that atoms are countable).

Write 
$$
\atomsdown_\nu=\{f_\nu(i)\mid i<0\}
\qquad
\atomsup_\nu=\{f_\nu(i)\mid i\geq 0\}.
$$
Finally, write 
$$
\atomsdown=\bigcup\atomsdown_\nu
\qquad
\atomsup=\bigcup\atomsup_\nu
\qquad
\mathbb A=\bigcup \mathbb A_\nu
$$ 
$a,b,c,\ldots$ will range over \emph{distinct} atoms (we call this the \deffont{permutative} convention).

A \deffont{permission set} has the form $(\atomsdown \cup A)\setminus B$ where $A\subseteq\atomsup$ and $B\subseteq\atomsdown$ are finite.
$S$, $T$, and $U$ will range over permissions sets. 
\end{defn}
The use of $\atomsdown$ and $\atomsup$ ensures that permission sets are infinite and also co-infinite (their complement with respect to $\mathbb A$ is also infinite).

\begin{rmrk}[Representing permission sets]
A permission set $S$ may be finitely represented 
\begin{itemize*}
\item
either as the pair of finite sets $(A,B)$ where $A\subseteq\atomsdown$ and $B\subseteq\atomsup$ and $S=(\atomsdown\setminus A)\cup B$, 
\item
or perhaps more elegantly as a single finite set $C\subseteq\mathbb A$ such that $S=\atomsdown\shs C$ where $X\shs Y=\{z\mid (z\in X\land z\not\in Y)\lor (z\not\in X\land z\in Y)$ (exclusive or). 
\end{itemize*}
Permission sets are a sorting/typing annotation which will be associated to variables in Definition~\ref{defn.terms.sorts}.
\end{rmrk}

\begin{frametxt}
\begin{defn}
\label{defn.permutation}
Given $a,b\in\mathbb A_\nu$ let a \deffont{(level 1) swapping} $(a\ b)$ be the bijection on atoms that maps $a$ to $b$, $b$ to $a$, and all other $c$ to themselves.

Also define a bijection $\f{shift}_\nu$ on atoms by: 
$$
\begin{array}{r@{\ }l@{\qquad}l}
\f{shift}_\nu(a)=&f_\nu(f\mone_\nu(a)+1)& (a\in\mathbb A_\nu) 
\\
\f{shift}_\nu(a)=&a & a\in\mathbb A\setminus\mathbb A_\nu
\end{array}
$$ 

Let the \deffont{(level 1) permutations} be the group of bijections on atoms generated by all swappings and $\f{shift}_\nu$.

Call a permutation $\pi$ \deffont{finite} when it is generated just by swappings; thus, when
$\f{nontriv}(\pi)=\{a{\in}\mathbb A\mid \pi(a){\neq} a\}$ is finite.
Otherwise, call $\pi$ \deffont{non-finite}.

$\pi$ will range over permutations.
Write $\mathbb P$ for the set of all permutations
and write $\mathbb P_{\text{fin}}$ for the set of all finite permutations. 
\end{defn}
\end{frametxt}

\begin{rmrk}
Swappings are used to manage $\alpha$-equivalence in nominal terms.
This is standard and goes back (at least) to \cite{gabbay:newaas-jv} and the second author's thesis \cite{gabbay:thesis}.

The true importance of $\f{shift}_\nu$ is that it bijects $\atomsdown$ with $\atomsdown\cup\{a\}$ for some $a\not\in\atomsdown$---this cannot be achieved using a \emph{finite} permutation.
The relevance of this is that later when we build $\Forall{X}\phi$, this really will mean `for all $X$' even though the permission set $S$ of $X$ makes it range over terms with free atoms in $S$. 

Permutations, like permission sets, easily admit finite representations.
$\f{shift}$ corresponds via the bijection with numbers to the operation `add 1'.
\end{rmrk}

\begin{defn}
\label{defn.terms.sorts}
For each signature $\mathcal S=(\mathcal A,\mathcal B,\mathcal F,\mathcal P,\f{ar})$ and each sort $\alpha$ over $(\mathcal A,\mathcal B)$ and permission set $S$ fix a countably infinite set of \deffont{unknowns} (\emph{level 2 names}) of that sort and permission set.
$X,Y,Z$ will range over distinct unknowns.
Write $\sort(X)$ for the sort and $\pmss(X)$ for the permission set of $X$.
\end{defn}

\begin{rmrk}
So an unknown $X$ has two type attributes: its \emph{sort} $\alpha$, which intuitively describes what kind of data it denotes, and its \emph{permission set} $S$ which describes the permitted free atoms of the terms, and also the nominal support of the denotations, with which it may be associated by a substitution or valuation---see the definitions of substitution and valuation in Definitions~\ref{defn.level.2.sub} and~\ref{defn.valuation} respectively.

If $X$ has sort `integers' and permission set $\atomsdown$, then intuitively $X$ represents `a term denoting an integer, with free atoms in $\atomsdown$'.

Another name for permission set $S$ might be \emph{freshness set}, since equally $X$ represents ``terms with free atoms \emph{not} in $\mathbb A\setminus S$''.\footnote{Via this intuition, permission sets correspond to the \emph{freshness constraints} $a\#X$ of \cite{gabbay:nomu-jv}.
For the reader familiar with freshness constraints, another way to view permission sets is as fixing a single global freshness context with `enough' freshnesses (the germ of this was already in \cite{gabbay:newcc}) where `enough' means that for any term, we can always pick an atom not free in that term.
However the implications of doing this go beyond a syntactic tweak to nominal terms; permission sets are what make it possible for us to reconcile level 2 quantification $\forall X$ with level 1 atoms-abstraction $[a]r$.} 
\end{rmrk}

\begin{defn}
\label{defn.syntax}
For each signature $\mathcal S$, define \deffont{raw terms} and \deffont{raw propositions} over $\mathcal S$ by: 
\begin{frameqn}
\begin{array}{c@{\qquad}c@{\qquad}c}
\begin{prooftree}
(a\in\mathbb A_\nu)
\justifies
a:\nu
\end{prooftree}
&
\begin{prooftree}
\rawr_1:\alpha_1 \ \ldots\ \rawr_n:\alpha_n
\justifies
(\rawr_1,\ldots,\rawr_n):(\alpha_1,\ldots,\alpha_n)
\end{prooftree}
&
\begin{prooftree}
\rawr:\alpha\quad (\f{ar}(\tf f)=(\alpha)\tau)
\justifies
\tf f(\rawr):\tau
\end{prooftree}
\\[4ex]
\begin{prooftree}
\rawr:\alpha\ (a\in\mathbb A_\nu)
\justifies
[a]\rawr:[\nu]\alpha
\end{prooftree}
&
\begin{prooftree}
(\sort(X)=\alpha)
\justifies
\pi\act X:\alpha
\end{prooftree}
\\[4ex]
\begin{prooftree}
\phantom{h}
\justifies
\bot\text{ prop.}
\end{prooftree}
&
\begin{prooftree}
\rawphi\text{ prop.}\ \ \rawpsi\text{ prop.}
\justifies
\rawphi\limp\rawpsi\text{ prop.}
\end{prooftree}
&
\begin{prooftree}
\rawr:\alpha\ \ (\f{ar}(\tf P)=\alpha)
\justifies
\tf P(\rawr)\text{ prop.}
\end{prooftree}
\\[4ex]
\begin{prooftree}
\rawphi\text{ prop.}
\justifies
\Forall{X}\rawphi\text{ prop.}
\end{prooftree}
\end{array}
\end{frameqn}
As in Definition~\ref{defn.sort.sig}, we admit the possiblity of empty tuples so that $()$ the empty tuple of terms is a term and has sort $()$. 
\end{defn}
We will quotient raw terms and propositions by $\alpha$-equivalence (to obtain terms $r$ and propositions $\phi$), later.

\begin{xmpl}
Consider $\tf{lam}([b]\tf{app}(X,\tf{var}(b)))$ where $b\not\in\pmss(X)$;
this represents the $\lambda$-term schema $\lam{y}(ty)$ where $y\not\in\f{fv}(t)$.

Recall that $\tf{app}$ and $\tf{lam}$ are term-formers of arities $(\iota,\iota)\iota$ and $([\nu]\iota)\iota$.
The sorts of $b$ and $X$ are $\nu$ (names) and $\iota$ (individuals) respectively. 
\end{xmpl}

\begin{rmrk}
\label{rmrk.no.exists}
Our version of PNL has connectives $\bot$, $\limp$, and $\forall$.
We could easily add other connectives like $\top$, $\neg$, $\land$, $\lor$, and $\exists$.
Instead we treat them as a definable extension using the standard `de Morgan' encoding. 
\end{rmrk}
We may write $\id\act X$ just as $X$.

\subsection{Permutation actions and free atoms/unknowns}

Nominal techniques suggest handling $\alpha$-renaming using permutations.
To a first approximation, if wherever the reader sees `permutation action' they substitute `$\alpha$-renaming', then they will not go too far wrong.

\begin{nttn}
\label{nttn.permutations}
We use the following notation:
\begin{itemize*}
\item
Write $\pi\circ\pi'$ for \deffont{functional composition}, so $(\pi\circ\pi')(a)=\pi(\pi'(a))$).
\item
Write $\id$ for the \deffont{identity permutation}, so $\id(a)=a$ always. 
\item
Write $\pi\mone$ for \deffont{inverse}, so $\pi\circ\pi\mone=\id$.
\item
Define $\pi^n$ by
$\pi^0=\id$
\ and\ 
$\pi^{n+1}=\pi^n\circ\pi$.
\end{itemize*}
\end{nttn}

\maketab{tab2}{R{7em}@{\ }L{6em}@{\ }R{8em}@{\ }L{10em}}
\begin{defn}
\label{defn.permutation.action}
Define a (level 1) \deffont{permutation action} on syntax by:
\begin{tab2}
\pi\act a\equiv& \pi(a)
&
\pi\act (\rawr_1,\ldots,\rawr_n) \equiv&  (\pi\act \rawr_1,\ldots,\pi\act \rawr_n)
\\
\pi\act [a]\rawr \equiv&  [\pi(a)]\pi\act \rawr
&
\pi\act(\pi'\act X) \equiv&  (\pi{\circ}\pi')\act X
\\
\pi\act \tf f(\rawr) \equiv&  \tf f(\pi\act \rawr)
\\
\pi\act\bot \equiv&  \bot
&
\pi\act (\rawphi\limp\rawpsi)\equiv&  (\pi\act \rawphi)\limp(\pi\act \rawpsi)
\\
\pi\act \tf P(\rawr)\equiv&  \tf P(\pi\act \rawr)
&
\pi\act (\Forall{X}\rawphi) \equiv&  \Forall{X}\pi\act\rawphi
\end{tab2}
\end{defn}

\begin{defn}
\label{defn.permutation.action.2}
Let $\Pi$ range over sort- and permission-set-preserving bijections on unknowns 
(so $\sort(\Pi(X)){=}\sort(X)$ and $\pmss(\Pi(X)){=}\pmss(X)$)
such that $\{X\mid \Pi(X)\neq X\}$ is finite.

Write $\Pi\circ\Pi'$ for functional composition,\ $\Id$ for the identity permutation, and $\Pi\mone$ for inverse, much as in Notation~\ref{nttn.permutations}.

Define a (level 2) \deffont{permutation action} by:
\begin{tab2}
\Pi\act a\equiv&  a
&
\Pi\act (\rawr_1,\ldots,\rawr_n) \equiv&  (\Pi\act \rawr_1,\ldots,\Pi\act \rawr_n)
\\
\Pi\act [a]\rawr \equiv&  [a]\Pi\act \rawr
&
\Pi\act(\pi\act X) \equiv&  \pi\act(\Pi(X))
\\
\Pi\act \tf f(\rawr) \equiv&  \tf f(\Pi\act \rawr)
\\
\Pi\act\bot \equiv&  \bot
&
\Pi\act (\rawphi\limp\rawpsi)\equiv&  (\Pi\act \rawphi)\limp(\Pi\act \rawpsi)
\\
\Pi\act \tf P(\rawr)\equiv&  \tf P(\Pi\act \rawr)
&
\Pi\act (\Forall{X}\rawphi) \equiv&  \Forall{\Pi(X)}\Pi\act\rawphi
\end{tab2}
\end{defn}

\subsection{Free level 1 and level 2 variables}

\begin{defn}
\label{defn.pointwise}
Suppose $A$ is a set of atoms and $\pi$ is a level 1 permutation.
Suppose $U$ is a set of unknowns and $\Pi$ is a level 2 permutation.
Define $\pi\act A$ and $\Pi\act U$ by
$$
\pi\act A = \{\pi(a)\mid a\in A\} \qquad\text{and}\qquad
\Pi\act U = \{\Pi(X)\mid X\in U\}.
$$
This is the standard \deffont{pointwise} permutation action on sets.
\end{defn}

\begin{defn}
\label{defn.fa}
Define \deffont{free atoms} $\f{fa}(\rawr)$ and $\f{fa}(\rawphi)$ by:
\begin{tab2}
\f{fa}(\pi\act X)=& \pi\act\pmss(X) 
&
\f{fa}([a]\rawr)=& \f{fa}(\rawr)\setminus\{a\}
\\
\f{fa}(\tf f(\rawr)) =&  \f{fa}(\rawr)
&
\f{fa}((\rawr_1,\ldots,\rawr_n)) =& 
\bigcup\f{fa}(\rawr_i)
\\
\f{fa}(a)=& \{a\}
&&
\\[1.5ex]
\f{fa}(\bot) =& \varnothing
&
\f{fa}(\rawphi\limp\rawpsi)=& \f{fa}(\rawphi)\cup \f{fa}(\rawpsi)
\\
\f{fa}(\tf P(\rawr)) =&  \f{fa}(\rawr)
&
\f{fa}(\Forall{X}\rawphi)=& \f{fa}(\rawphi) 
\end{tab2}
Define \deffont{free unknowns} $\f{fV}(r)$ and $\f{fV}(\rawphi)$ by:
\begin{tab2}
\f{fV}(a)=& \varnothing
&
\f{fV}(\pi\act X)=& \{X\}
\\
\f{fV}([a]\rawr)=& \f{fV}(\rawr)
&
\f{fV}((\rawr_1,\ldots,\rawr_n)) =&  
\bigcup\f{fV}(\rawr_i)
\\
\f{fV}(\tf f(\rawr)) =&  \f{fV}(\rawr)
&&
\\[1.5ex]
\f{fV}(\bot) =& \varnothing
&
\f{fV}(\rawphi\limp\rawpsi)=& \f{fV}(\rawphi)\cup \f{fV}(\rawpsi)
\\
\f{fV}(\tf P(\rawr)) =&  \f{fV}(\rawr)
&
\f{fV}(\Forall{X}\rawphi)=& \f{fV}(\rawphi)\setminus\{X\} 
\end{tab2}
\end{defn}

\begin{lemm}
\label{lemm.fa.pi.r}
$\f{fa}(\pi\act \rawr)=\pi\act \f{fa}(\rawr)$ and $\f{fa}(\pi\act\rawphi)=\pi\act\f{fa}(\rawphi)$.

Also,
$\f{fV}(\Pi\act \rawr)=\Pi\act \f{fV}(\rawr)$ and $\f{fV}(\Pi\act\rawphi)=\Pi\act\f{fV}(\rawphi)$.
\end{lemm}
\begin{proof}
By routine inductions on $\rawr$.
\end{proof}

\subsection{$\alpha$-equivalence}

\begin{defn}
Call an equivalence relation $\somerel$ on terms and on propositions a \deffont{congruence} when it is closed under the following rules:
$$
\begin{array}{c@{\qquad\quad}c}
\begin{prooftree}
\rawr_i\somerel \raws_i\quad 1\leq i\leq n
\justifies
(\rawr_1,\ldots,\rawr_n)\somerel (\raws_1,\ldots,\raws_n)
\end{prooftree}
&
\begin{prooftree}
\rawr\somerel \raws\ \ (\tf f:(\alpha)\tau,\ \rawr,\raws:\alpha)
\justifies
\tf f(\rawr)\somerel\tf f(\raws)
\end{prooftree}
\\[3ex]
\begin{prooftree}
\rawr\somerel \raws
\justifies
[a]\rawr\somerel [a]\raws
\end{prooftree}
&
\begin{prooftree}
\rawphi\somerel\rawphi'\quad \rawpsi\somerel\rawpsi'
\justifies
\rawphi\limp\rawpsi\somerel \rawphi'\limp\rawpsi'
\end{prooftree}
\\[3ex]
\begin{prooftree}
\rawr\somerel \raws\quad (\tf P:\alpha,\ \rawr,\raws:\alpha)
\justifies
\tf P(\rawr)\somerel \tf P(\raws)
\end{prooftree}
&
\begin{prooftree}
\rawphi\somerel \rawphi'
\justifies
\Forall{X}\rawphi\somerel \Forall{X}\rawphi'
\end{prooftree}
\end{array}
$$
\end{defn}

\begin{defn}
\label{defn.aeq}
Write $(a\ b)$ for the \deffont{(level 1) swapping} permutation which maps $a$ to $b$, $b$ to $a$, and all other $c$ to themselves.
Similarly write $(X\ Y)$ for the \deffont{(level 2) swapping}.

Define \deffont{$\alpha$-equivalence} on terms and propositions to be the least congruence $\aeq$ such that:
\begin{frameqn}
\begin{array}{c@{\qquad}c}
\begin{prooftree}
(b\ a)\act \rawr\aeq \raws \quad (b\not\in\f{fa}(\rawr))
\justifies
[a]\rawr\aeq [b]\raws
\end{prooftree}
&
\begin{prooftree}
(\pi(a)=\pi'(a)\ \text{ for all }\ a{\in}\pmss(X))
\justifies
\pi\act X\aeq \pi'\act X
\end{prooftree}
\\[4ex]
\begin{prooftree}
(Y\ X)\act \rawphi\aeq \rawpsi \quad (Y\not\in\f{fV}(\rawphi))
\justifies
\Forall{X}\rawphi\aeq\Forall{Y}\rawpsi
\end{prooftree}
\end{array}
\end{frameqn}
\end{defn}

\begin{xmpl}
We $\alpha$-convert $X$ and $a$ in $\Forall{X}\tf P([a]X)$.

Let $\sort(Y)=\sort(X)$ and $\pmss(Y)=\pmss(X)=\atomsdown$.
Suppose $a\in\atomsdown$ and $b\not\in\atomsdown$.
Using $(a\ b)$ and $(X\ Y)$ we deduce:
$$
\begin{array}{r@{\quad}c@{\quad}l}
\Forall{X}\tf P([a]X) &\stackrel{(a\ b)}{\aeq}& \Forall{X}\tf P([b](b\ a)\act X) 
\\
&\stackrel{(X\ Y)}{\aeq}& \Forall{Y}\tf P([b](b\ a)\act Y) . 
\end{array}
$$  
It is routine to convert this sketch into a full derivation-tree.

Furthermore, if we take syntax as above except that $a\not\in\atomsdown$ and $b\not\in\atomsdown$, then we deduce $\Forall{X}\tf P([a]X)\aeq \Forall{Y}\tf P([b]Y)$.
\end{xmpl}

\begin{rmrk}
Note that $\alpha$-equivalence is
highly symmetric between levels 1 and 2,\ based on permutations instead of substitutions, and avoids equality reasoning in the logic.

In \cite{gabbay:newaas-jv,pitts:nomlfo-jv,gabbay:frelog,gabbay:seqcnl,cheney:simptn} it is not in general possible to `just $\alpha$-convert' a level 1 abstraction.
We must appeal instead to equality reasoning describing atoms-abstraction in nominal sets.
But this is harder; derivable equality is more complex than syntactic equivalence. 
\end{rmrk}

\begin{lemm}
\label{lemm.aeq.equivar}
For every $\pi$, $\Pi$, $\rawr$, $\raws$, $\rawphi$, and $\rawpsi$, the following hold:
\begin{itemize*}
\item
$\rawr\aeq \raws$ if and only if $\pi\act \rawr\aeq \pi\act \raws$ and similarly $\rawphi\aeq\rawpsi$ if and only if $\pi\act\rawphi\aeq\pi\act\rawpsi$.
\item
$\rawr\aeq \raws$ if and only if $\Pi\act \rawr\aeq\Pi\act \raws$, and similarly $\rawphi\aeq\rawpsi$ if and only if $\Pi\act\rawphi\aeq\Pi\act\rawpsi$.
\end{itemize*}
\end{lemm}

\begin{lemm}
\label{lemm.fa.aeq}
If $\rawr\aeq \raws$ then $\f{fa}(\rawr)=\f{fa}(\raws)$ and $\f{fV}(\rawr)=\f{fV}(\raws)$.
\end{lemm}

\begin{prop}
\label{prop.aeq.equivalence}
$\aeq$ is an equivalence relation on terms and propositions.
\end{prop}
\begin{proof}
By a standard argument as in \cite{gabbay:nomr-jv}, using 
Lemmas~\ref{lemm.fa.pi.r},\ \ref{lemm.aeq.equivar},\ and~\ref{lemm.fa.aeq}.
\end{proof}

\begin{lemm}
\label{lemm.fa.fix}
If $\pi(a)=a$ for every $a\in\f{fa}(\rawr)$
then $\pi\act \rawr\aeq \rawr$. 
\end{lemm}

\begin{frametxt}
\begin{defn}
\label{defn.terms.and.propositions}
For each signature $\mathcal S$, define \deffont{terms} and \deffont{propositions} over $\mathcal S$ to be raw terms and propositions quotiented by $\alpha$-equivalence.

$r$ and $s$ will range over terms.
$\phi$ and $\psi$ will range over propositions.
\end{defn}
\end{frametxt}

\begin{rmrk}
Terms and propositions inherit the definitions and properties of raw terms and propositions.
Thus for example we may write $\f{fa}(r)$ to mean `$\f{fa}(\rawr)$ for some $\rawr\in r$' (Lemma~\ref{lemm.fa.aeq} proves this is well-defined).

It is possible to construct terms and propositions directly using a variant of nominal syntax-with-binding from \cite{gabbay:newaas-jv}, tweaked to include permutation and abstraction by unknowns.

It is also possible to retain the definitions above---reasoning on $\alpha$-equivalence classes of terms---and to use theorems of \emph{abstractive functions} developed in \cite{gabbay:genmn}, which are an alternative `nominal' way to guarantee well-definedness of functions defined on $\alpha$-equivalence classes.

We will not dwell on these issues in this paper because we are most interested in what PNL syntax can express rather than thinking about the syntax for its own sake.
However, the mathematics to do this exists and is well-understood.
\end{rmrk}

\subsection{Substitution}

\begin{frametxt}
\begin{defn}
\label{defn.level.2.sub}
A (level 2) \deffont{substitution} $\theta$ is a function from unknowns to terms such that:
\begin{itemize*}
\item
For all $X$, $\theta(X):\sort(X)$ and $\f{fa}(\theta(X))\subseteq \pmss(X)$.
\item
$\theta(X)\equiv \id\act X$ for all but finitely many $X$.
\end{itemize*}
$\theta$ will range over substitutions.
\end{defn}
\end{frametxt}

One kind of substitution will be particularly useful later:
\begin{defn}
\label{defn.atomic.sub}
Suppose $X$ is an unknown and suppose $t:\sort(X)$ and $\f{fa}(t)\subseteq\pmss(X)$.
Define $[X\ssm t]$ by:
\begin{frameqn}
\begin{array}{l@{\ }l@{\qquad}l}
{[X\ssm t]}(X)=&t
\\
{[X\ssm t]}(Y)=&\id\act Y &\text{all other }Y
\end{array}
\end{frameqn}
\end{defn}
 
By convention (Definition~\ref{defn.terms.sorts}) $X$ and $Y$ in Definition~\ref{defn.nontriv.theta} range over \emph{distinct} unknowns: 
\begin{defn}
\label{defn.nontriv.theta}
Define $\f{nontriv}(\theta)$ by:
$$
\f{nontriv}(\theta)\equiv 
\{X\mid \theta(X){\not\equiv} \id\act X \text{ or } X{\in}\f{fV}(\theta(Y))\text{ for some }Y\} 
$$
\end{defn}
$\f{nontriv}(\theta)$ is unknowns that can be produced or consumed by $\theta$, other than in the trivial manner that $\theta(X)\equiv\id\act X$.

\begin{defn}
\label{defn.subst.action}
Define a \deffont{substitution action} by:
\begin{frameqn}
\begin{array}{r@{\ }l@{\qquad}r@{\ }l}
a\theta\equiv&  a
&
(r_1,\ldots,r_n)\theta\equiv&  (r_1\theta,\ldots,r_n\theta)
\\
([a]r)\theta\equiv&  [a](r\theta)
&
(\pi\act X)\theta\equiv&  \pi\act \theta(X)
\\
\tf f(r)\theta\equiv&  \tf f(r\theta)
\\
\bot\theta\equiv& \bot
&
(\phi\limp\psi)\theta\equiv&  (\phi\theta)\limp\psi\theta
\\
(\tf P(r))\theta\equiv&  \tf P(r\theta)
&
(\Forall{X}\phi)\theta \equiv&  \Forall{X}(\phi\theta) \quad (X\not\in\f{nontriv}(\theta))
\end{array}
\end{frameqn}
\end{defn}

\begin{rmrk}
Level 2 substitution $r\theta$ is capturing for level 1 abstraction $[a]\text{-}$.
For example if $\theta(X)=a$ then $([a]X)\theta\equiv [a]a$.
This is the behaviour displayed by the informal meta-level when we write ``take $t$ to be $x$ in $\lam{x}t$''.

Only atoms in $\pmss(X)$ may be captured in this way.
Thus for instance, if $a\not\in\pmss(X)$ then $\theta(X)=a$ is impossible because it would violate the condition $\f{fa}(\theta(X))\subseteq\pmss(X)$ in Definition~\ref{defn.level.2.sub}.
\end{rmrk}

\begin{rmrk}
\label{rmrk.next.remark}
The condition $\f{fa}(\theta(X))\subseteq\pmss(X)$ in Definition~\ref{defn.level.2.sub}, and the condition $\f{fa}(t)\subseteq\pmss(X)$ in Definition~\ref{defn.atomic.sub}, are necessary for the substitution action in Definition~\ref{defn.subst.action} to be well-defined.

Consider a name sort $\nu$ and suppose $X:\nu$ and $a,b:\nu$.
Suppose $a,b\not\in\pmss(X)$, so that by Definition~\ref{defn.aeq} $(a\ b)\act X=\id\act X$.\footnote{Remember that we quotient raw terms by $\alpha$-equivalence to obtain terms so this is now a real equality.}

Suppose we drop the conditions on free atoms of terms, so that we admit $[X\ssm a]$ as a substitution. 
Then according to the definitions, $((a\ b)\act X)[X\ssm a]=b$ whereas $(\id\act X)[X\ssm a]=a$. 
\end{rmrk}

\begin{rmrk}
\label{rmrk.atoms.as.data}
In PNL, atoms are data; they are `bindable constant symbols'.
Atoms are not variables; they do not come with a substitution as a primitive in PNL. 
(Unknowns are variables; they have a substitution action.)

The reader should not expect atoms to populate every sort, like variables do.
Atoms populate their own special sorts, name-sorts, which are sorts for `bindable data'.

We can make atoms populate a base sort (e.g. variable symbols with sort `$\lambda$-terms' or `functions') with a term-former (Definition~\ref{defn.term.signature}), e.g. $\tf{var}$ in Sections~\ref{sect.pnl.arithmetic} and~\ref{sect.more.examples}.

We can make atoms behave like variables using axioms, like those of \theory{SUB} in Figure~\ref{fig.substitution}. 
In PNL, a substitution action for atoms is a matter of writing suitable axioms.
Fortunately, nominal techniques make this fairly easy to do.
\end{rmrk}

\subsection{Sequents and derivability}

\begin{defn}
\label{defn.seq}
$\Phi$ and $\Psi$ will range over sets of propositions.
We may write $\phi,\Phi$ and $\Phi,\phi$ as shorthand for $\{\phi\}\cup\Phi$. 
We may write $\Phi,\Psi$ as shorthand for $\Phi\cup\Psi$.
Finally, define $\f{fV}(\phi)$ by
$$
\f{fV}(\Phi)=\bigcup\{\f{fV}(\phi)\mid \phi\in\Phi\}.
$$ 
A \deffont{sequent} is a pair $\Phi\cent\Psi$.
\end{defn}

\begin{frametxt}
\begin{defn}[Derivable sequents]
The \deffont{derivable sequents} are defined in Figure~\ref{Seq}. 
\end{defn}
\end{frametxt}

\begin{figure*}[t]
$$
\begin{array}{c@{\qquad}c}
\begin{prooftree}
\phantom{h}
\justifies
\Phi,\,\phi\cent \pi\act\phi,\,\Psi
\using\rulefont{Ax}
\end{prooftree}
&
\begin{prooftree}
\phantom{h}
\justifies
\Phi,\,\bot\cent \Psi
\using\rulefont{\bot L}
\end{prooftree}
\\[4ex]
\begin{prooftree}
\Phi\cent \phi,\,\Psi
\quad
\Phi,\,\psi\cent \Psi
\justifies
\Phi,\,\phi\limp\psi\cent\Psi
\using\rulefont{{\limp}L}
\end{prooftree}
&
\begin{prooftree}
\Phi,\,\phi\cent \psi,\,\Psi
\justifies
\Phi\cent \phi\limp\psi,\,\Psi
\using\rulefont{{\limp}R}
\end{prooftree}
\\[4.2ex]
\begin{prooftree}
\raisebox{.8ex}{
$\begin{array}{c}
\Phi,\,\phi[X\ssm r]\cent \Psi
\\
(\f{fa}(r){\subseteq}\pmss(X), 
\ r{:}\sort(X))
\end{array}
$}
\justifies
\Phi,\,\Forall{X}\phi\cent \Psi
\using\rulefont{{\forall}L}
\end{prooftree}
&
\begin{prooftree}
\Phi\cent \phi,\,\Psi\quad {\small (X\not\in\f{fV}(\Phi,\Psi))}
\justifies
\Phi\cent \Forall{X}\phi,\,\Psi
\using\rulefont{{\forall}R}
\end{prooftree}
\\[5ex]
\begin{prooftree}
\Phi\cent \phi,\,\Psi
\qquad
\Phi,\,\phi\cent\Psi
\justifies
\Phi\cent\Psi
\using\rulefont{Cut}
\end{prooftree}
\end{array}
$$
\caption{Sequent derivation rules of Permissive-Nominal Logic}
\label{Seq}
\end{figure*}

\begin{rmrk}
As standard, the intuition of $\Phi\cent\Psi$ being derivable is ``the conjunction (logical and) of the propositions in $\Phi$ entails the disjunction (logical or) of the propositions in $\Psi$''.
So for instance, intuitively the axiom rule \rulefont{Ax} expresses that $\phi$ if and only if $\pi\act\phi$. 

The $\pi$ in \rulefont{Ax} is deliberate and represents equivariance (preservation of truth under permuting atoms) 
within permissive-nominal logic.\footnote{In this paper equivariance surfaces in a variety of technical features of PNL: the $\pi$ in \rulefont{Ax}; the way permutations distribute into terms in Definition~\ref{defn.permutation.action}; Definition~\ref{defn.finite.equivariant} and how it is used in Definition~\ref{defn.pnl.interp}; Lemma~\ref{lemm.pi.r.model}; and more.

In fact, equivariance is broad and useful phenomenon.  The interested reader is referred to \cite[Lemma~4.7]{gabbay:newaas-jv} and \cite[Subsection~4.2]{gabbay:fountl}, where it is treated in full generality.
See also \cite[Lemma~5.5]{gabbay:twolns}, where the conditions on free atoms and support in Definitions~\ref{defn.level.2.sub} and~\ref{defn.valuation} are also exhibited as forms of equivariance.
} 
Examples of how $\pi$ is used follow immediately in Subsection~\ref{subsect.all.terms}.
\end{rmrk}

\begin{nttn}
We may write $\Phi\cent\Psi$ as shorthand for `$\Phi\cent\Psi$ is a derivable sequent'.

We may write $\Phi\not\cent\Psi$ as shorthand for `$\Phi\cent\Psi$ is not a derivable sequent'.
\end{nttn}

\begin{rmrk}
Figure~\ref{Seq} includes rules for $\bot$, $\limp$, and $\forall$.
Following on from Remark~\ref{rmrk.no.exists}, note that including rules for other connectives like $\top$, $\neg$, $\land$, $\lor$, and $\exists$ would be easy.
Because the PNL in this paper is classical, we can treat derivation rules for them as a definable extension of what we already have.

We see no inherent difficulty with constructing an intuitionistic version of PNL.
\end{rmrk}

\subsection{Universal quantification, permission sorts, and $\f{shift}$-permutations}
\label{subsect.all.terms}

Recall the comment on `atoms as data' in Remark~\ref{rmrk.atoms.as.data}.
Because of permutations, in certain circumstances free atoms can still behave like variables ranging over distinct atoms (cf. the \emph{permutative convention} of Definition~\ref{defn.atoms}).
Atoms-substitution is not be primitive in PNL, but atoms-permutation is. 

Thus in PNL we can express a theory of atoms-inequality in the following interesting way:
Assume a name sort $\tf{Atm}$ and 
a proposition-former $\tf{neq}:(\tf{Atm},\tf{Atm})$, along with a single proposition $\tf{neq}(a,b)$ for two distinct atoms in $\tf{Atm}$---and, if we wish, a proposition $\tf{neq}(a,a)\limp\bot$.
The permutation $\pi$ in \rulefont{Ax} ensures that $a$ and $b$ represent \emph{any} two distinct atoms. 

This goes further.
The condition $\f{fa}(r)\subseteq\pmss(X)$ in \rulefont{\forall L} might suggest that $\Forall{X}\phi$ means ``$\phi[X\ssm r]$ for every $r$ such that $\f{fa}(r)\subseteq\pmss(X)$''.
Indeed this is so, but what $\pmss(X)$ in $\Forall{X}\phi$ really restricts is \emph{capture}, as we now discuss.
\begin{itemize}
\item
Suppose a name sort $\tf{Atm}$ and suppose $X:\tf{Atm}$ and a proposition-former $\tf P$ of arity $\tf{Atm}$.
Suppose $b\in\f{pmss}(X)$.
By considering the swapping $(b\ a)$ and \rulefont{Ax}, and \rulefont{{\forall}L}, 
$\Forall{X}\tf P(X) \cent \tf P(a)$ for \emph{all} $a$, even if $a\not\in\pmss(X)$, as follows: 
$$
\begin{prooftree}
\begin{prooftree}
\justifies
\tf P(b) \cent \tf P(a)
\using\rulefont{Ax}\quad \text{$\pi=(b\ a)$} 
\end{prooftree}
\justifies
\Forall{X}\tf P(X)\cent \tf P(a)
\using\rulefont{\forall L}\quad\text{$[X\ssm b]$}
\end{prooftree}
$$
In other words, we can derive $\tf P(a)$ from $\Forall{X}\tf P(X)$, even if $a$ is not permitted in $X$.

Thus, in the case of level 2 closed terms (without unknowns), these have finitely many atoms and we can use a finite permutation to place them in $\pmss(X)$.
\item
This may not work for the more general case of a term \emph{with} unknowns; for example there is no finite $\pi$ such that $\f{fa}(\pi\act(X,a))\subseteq\pmss(X)$ where $a\not\in\pmss(X)$.

So consider the general case of any sort $\alpha$ and suppose $X:\alpha$ and $\pmss(X)=S$.
Suppose $\tf Q:\alpha$.
Consider any other $Y:\alpha$ and $\pmss(Y)=T$.
We will show that $\Forall{X}\tf Q(X)\cent \tf Q(Y)$ is derivable.

Recall that shift permutation $\f{shift}_\nu$ from Definition~\ref{defn.permutation} and  the definition of $\f{shift}^n_\nu$ from Notation~\ref{nttn.permutations}.
Using $\f{shift}$ permutations we can construct a permutation $\pi$ such that $S\cup T\subseteq \pi\act\f{pmss}(X)$.\footnote{If permission sets were finite, or if all permutations were finite, then we could not do this in general.}

We derive as follows:
$$
\begin{prooftree}
\begin{prooftree}
\justifies
\tf Q(\pi\mone\act Y)\cent \tf Q(Y)
\using\rulefont{Ax} 
\end{prooftree}
\justifies
\Forall{X}\tf Q(X) \cent \tf Q(Y)
\using \rulefont{{\forall}L}\quad\text{$[X\ssm \pi\mone\act Y]$}
\end{prooftree}
$$
\item
Nevertheless, $\Forall{X}\phi$ does not mean ``$\phi[X\ssm r]$ for every $r$''.
This is because permutations are bijective.
For example, suppose $X:\tf{Atm}$,\ $a\not\in\pmss(X)$, and $\tf P:([\tf{Atm}]\tf{Atm})$. 
Then $\Forall{X}\tf P([a]X) \cent P([a]r)$ for all $r$ such that $a\not\in\f{fa}(r)$, and also $\Forall{X}\tf P([b]X)\cent P([b]r)$ for all $r$ and all $b$ such that $b\not\in\f{fa}(r)$.
However, 
$$
\Forall{X}\tf P([a]X)\not\cent P([a]a), \quad\text{and for all $b$,}\ \ 
\Forall{X}\tf P([a]X)\not\cent P([b]b).
$$
The fact that $a\not\in\pmss(X)$ forbids capture by an instantiation, in a suitable sense.
\end{itemize}

\section{Semantics of permissive-nominal logic}
\label{sect.permissive-nominal.sets}

Nominal sets were introduced in \cite{gabbay:newaas-jv} (they were called `FM sets').
Technically, a nominal set is a set with a finitely-supported permutation action for atoms.
Intuitively, a nominal set is a set with `free names' in a manner which parallels how names feature in abstract syntax, but without necessarily being syntactic structures.

Names in nominal sets are modelled as atoms.
They can be renamed by the permutation action; they can be bound by an atoms-abstraction construction; and they feature a \emph{finite support} property which guarantees that we can always pick a fresh name. 

The interested reader can consult a literature which includes \cite{gabbay:newaas-jv} or \cite{gabbay:fountl} for detailed discussions of nominal sets and their applications. 
 
We will interpret PNL using \emph{permissive-nominal} sets. 
The permissive-nominal sets we use here generalise nominal sets in two ways: 
\begin{itemize*}
\item
They allow \emph{infinite support}, since permission sets from Definition~\ref{defn.atoms} need not be finite.
\item
They allow (some) \emph{infinite permutations}, since permutations from Definition~\ref{defn.permutation} are generated as a group by swappings and by infinite $\f{shift}$ permutations, thus giving $\forall$-quantification extra power as discussed in Subsection~\ref{subsect.all.terms}.
\end{itemize*}

The main results are soundness (Theorem~\ref{thrm.soundness}) and completeness (Theorem~\ref{thrm.completeness}).
The main definitions are of permissive-nominal sets in Definition~\ref{defn.nominal.set}, and of the interpretations of terms and propositions in Definitions~\ref{defn.interpret.terms} and~\ref{defn.truth} respectively.

The permissive-nominal development here resembles that in \cite{gabbay:newaas-jv}.
Definition~\ref{defn.nominal.set} is a little subtle because we ignore infinite permutations when we determine support, whereas equivariance from Definition~\ref{defn.finite.equivariant} means commuting with \emph{all} permutations.

\subsection{Permissive-nominal sets}

Recall $\mathbb P$ the set of all permutations from Definition~\ref{defn.permutation}.
\begin{frametxt}
\begin{defn}
\label{defn.perm.set} 
A \deffont{set with a permutation action} $\ns X$ is a pair $(|\ns X|,\act)$ of 
\begin{itemize*}
\item
a \deffont{carrier set} $|\ns X|$ and 
\item
a group action on the carrier set $(\mathbb P\times |\ns X|)\to |\ns X|$, written infix as $\pi\act x$.

So, $\id\act x=x$ and $\pi\act(\pi'\act x)=(\pi\circ\pi')\act x$ for every $\pi$ and $\pi'$ and every $x\in|\ns X|$.
\end{itemize*}
\end{defn}
\end{frametxt}

\begin{defn}
\label{defn.support}
Say a set of atoms $A\subseteq\mathbb A$ \deffont{supports} $x\in |\ns X|$ when for all finite permutations $\pi$, if $\pi(a) =a$ for all $a \in A$ then $\pi\act x =x$.
\end{defn}
Thus, if a permission set $S$ supports $x$ and $\Forall{a{\in} S}\pi(a)=\pi'(a)$ then $\pi\act x=\pi'\act x$.

\begin{rmrk}
\label{rmrk.fuzzy.support}
$\mathbb P$ contains infinite as well as finite permutations.
In the next paragraph we construct an $x$ that is supported by $\varnothing$ (so that $\pi\act x=x$ for all finite $\pi\in\mathbb P$) and yet $\f{shift}_\nu(x)\neq x$. 
This observation is the same as \emph{fuzzy support} from \cite{gabbay:genmn}.

Recall the order $f_\nu$ on $\mathbb A_\nu$ from Definition~\ref{defn.atoms} and consider 
$$
x=\{\pi\act (f_\nu(0),f_\nu(\text{-}1),f_\nu(1),f_\nu(\text{-}2),\ldots)\mid \pi\in\mathbb P_{\text{fin}}\}
$$ 
the set of finite permutations of $\mathbb A_\nu$ written out in order, with the pointwise permutation action.
Then $\pi\act x=x$ for every finite permutation, but $\f{shift}_\nu\act x\neq x$.
\end{rmrk}

\begin{frametxt}
\begin{defn}
\label{defn.nominal.set} 
A \deffont{permissive-nominal set} is a set with a permutation action such that every element has a supporting permission set. 

$\ns X$, $\ns Y$ will range over permissive-nominal sets.
\end{defn}
\end{frametxt}

\begin{thrm}
\label{thrm.supp}
Suppose $\ns X$ is a permissive-nominal set.
Then every $x\in|\ns X|$ has a unique least supporting set $\f{supp}(x)\subseteq\mathbb A$.\footnote{$\f{supp}(x)$ need not necessarily be a permission set.  For instance, $\f{supp}(a)=\{a\}$.}

As a corollary, if $\pi$ is finite and $\pi(a)=a$ for all $a\in\f{supp}(x)$ then $\pi\act x=x$.
\end{thrm}
Theorem~\ref{thrm.supp} is familiar from \cite{gabbay:newaas-jv}, but we do have to be a little bit careful since we are not working with nominal sets.
It all works out:
\begin{proof}
The corollary is immediate given the definition of support (Definition~\ref{defn.support}).

Define $A=\bigcap\{S\mid S\text{ permission set, supports }x\}$.
Also, choose some permission set $S$ that supports $x$.

Suppose $\pi$ is finite and $\pi(a)=a$ for all $a\in A$.
Write $a_1,\ldots,a_n$ for the atoms in $\f{nontriv}(\pi)\cap S$, in some order.
Let $b_1,\ldots,b_n$ be some choice of fresh atoms (so $b_i\not\in S\cup\f{nontriv}(\pi)\cup A$ for $1\leq i\leq n$).
Write $\tau=(b_1\ a_1)\circ\ldots\circ(b_n\ a_n)$.
It is routine to check that $(\tau\circ\pi\circ\tau)(a)=a$ for every $a\in S$.
Thus $\tau\act(\pi\act(\tau\act x))=x$.
Now $\tau\act x=x$, and it follows by a routine manipulation that $\pi\act x=x$ as required.
\end{proof}

\begin{lemm}
\label{lemm.supp.pi.x}
Suppose $\ns X$ is a permissive-nominal set and $x\in |\ns X|$.
Then $\f{supp}(\pi\act x)=\pi\act\f{supp}(x)$ (Definition~\ref{defn.pointwise}).
\end{lemm}
\begin{proof}
By a routine calculation using the group action. 
\end{proof}

\begin{corr}
\label{corr.notinsupp}
Suppose $\ns X$ is a permissive-nominal set and $x\in |\ns X|$.
Suppose $b\not\in\f{supp}(x)$.
Then $(b\ a)\act x= x$ implies $a\not\in\f{supp}(x)$.
\end{corr}
\begin{proof}
Suppose $b\not\in\f{supp}(x)$.
We prove the contrapositive.
Suppose $a\in\f{supp}(x)$.
By Lemma~\ref{lemm.supp.pi.x} $\f{supp}((b\ a)\act x)=(b\ a)\act \f{supp}(x)$.
By our suppositions, $(b\ a)\act\f{supp}(x)\neq \f{supp}(x)$.
It follows that $(b\ a)\act x\neq x$. 
\end{proof}

\subsection{Examples of permissive-nominal sets}

\subsubsection{Atoms}

\begin{defn}
\label{defn.atoms.perm}
$\mathbb A$ the set of atoms can be considered a permissive-nominal set with a natural permutation action $\pi\act a=\pi(a)$.

The set $\{0,1\}$ can be considered a permissive-nominal set with the natural \deffont{trivial} permutation action $\pi\act x=x$ for all $\pi\in\mathbb P$ and $x\in\{0,1\}$.

In the cases of $\mathbb A$ and $\{0,1\}$ only, we will be lax about the distinction between the set, and the permissive-nominal set with its natural permutation action.
\end{defn}

\subsubsection{Atoms-abstraction}

\begin{defn}
\label{defn.abstraction.sets}
Suppose $\ns X$ is a permissive-nominal set and $\mathbb A_\nu$ is a set of atoms.
Suppose $x\in|\ns X|$ and $a\in\mathbb A_\nu$.
Define \deffont{atoms-abstraction} $[a]x$ and $[\mathbb A_\nu]\ns X$ by:
\begin{frameqn}
\begin{array}{r@{\ }l}
[a]x =& \{(a,x)\}\cup \{(b,(b\ a)\act x)\mid b\in\mathbb A_\nu{\setminus} \f{supp}(x)\}
\\
|[\mathbb A_\nu]\ns X| =& \{[a]x\mid a\in\mathbb A_\nu,\ x\in|\ns X|\} 
\\
\pi\act [a]x =& [\pi(a)]\pi\act x
\end{array}
\end{frameqn}
\end{defn}

\begin{rmrk}
In the definition of $[a]x$ in Definition~\ref{defn.abstraction.sets} recall that by our permutative convention $b\neq a$.
An equivalent and more compact way of writing this is $[a]x=\{(\pi(a),\pi\act x)\mid \pi\in\f{fix}(\f{supp}(x){\setminus}\{a\})\}$ where $\f{fix}(A)=\{\pi\mid\Forall{a{\in}A}\pi(a)=a\}$.
\end{rmrk}

\begin{lemm}
\label{lemm.supp.abstraction}
\begin{enumerate*}
\item
$[\mathbb A_\nu]\ns X$ is a permissive-nominal set.
\item
$[a]x{=}[a]x'$ if and only if $x{=}x'$, for $a{\in}\mathbb A_\nu$ and $x{\in} |\ns X|$.
\item
$[a]x{=}[a']x'$ if and only if $a'{\not\in}\f{supp}(x)$ and $(a'\, a)\act x{=}x'$, for $a,a'{\in}\mathbb A_\nu$ and $x,x'{\in}|\ns X|$.
\end{enumerate*}
\end{lemm}

\subsubsection{Product}

\begin{defn}
\label{defn.times}
If $\ns X_i$ are permissive-nominal sets for $1\leq i\leq n$ then define $\ns X_1\times\ldots\times \ns X_n$ by:
$$
\begin{array}{r@{\ }l}
|\ns X_1\times\ldots\times\ns X_n|=&|\ns X_1|\times\ldots\times|\ns X_n|
\\
\pi\act (x_1,\ldots,x_n)=&(\pi\act x_1,\ldots,\pi\act x_n)
\end{array}
$$
\end{defn}

\begin{lemm}
\label{lemm.properties.of.support}
\begin{itemize*}
\item
$\f{supp}(a)=\{a\}$.
\item
$\f{supp}([a]x)=\f{supp}(x)\setminus\{a\}$.
\item
$\f{supp}((x_1,\ldots,x_n))=\bigcup\{\f{supp}(x_i)\mid 1\leq i\leq n\}$.
\end{itemize*}
\end{lemm}
\begin{proof}
Proofs are as in \cite{gabbay:newaas-jv} or \cite{gabbay:fountl}. 
\end{proof}

\subsection{Interpretation and soundness}
\label{subsect.semantics}

\subsubsection{Interpretation of signatures} 
 
\begin{defn}
\label{defn.interpretation}
Suppose $(\mathcal A,\mathcal B)$ is a sort-signature (Definition~\ref{defn.sort.sig}).

A \deffont{(PNL) interpretation} or \deffont{model} $\mathcal I$ for $(\mathcal A,\mathcal B)$ consists of an assignment of a nonempty permissive-nominal set $\basesort^\iden$ to each $\basesort\in\mathcal B$.\footnote{We favour the word `interpretation' for assigning a denotational interpretation to a logic, and `model' for checking whether the interpretation makes a theory (a set of axioms).  These senses overlap, in that an interpretation is a model for the empty theory.  In practice, we tend to use whichever word seems most appropriate in context.} 

We extend an interpretation $\mathcal I$ to sorts by:
\begin{frameqn}
\begin{array}{r@{\ }l@{\qquad}r@{\ }l}
\model{\basesort}=&\basesort^\iden
&
\model{(\alpha_1,\ldots,\alpha_n)}=&\model{\alpha_1}\times\ldots\times\model{\alpha_n}
\\
\model{\nu}=&\mathbb A_\nu
&
\model{[\nu]\alpha}=&[\mathbb A_\nu]\model{\alpha}
\end{array}
\end{frameqn}
\end{defn}

\begin{rmrk}
Note in Definition~\ref{defn.interpretation} that a base sort $\basesort$ is interpreted by a permissive-nominal set $\basesort^\iden$ given in the interpretation, whereas a name sort $\nu$ \emph{must} be interpreted by its corresponding set of atoms $\atoms_\nu$ as fixed in Definition~\ref{defn.atoms}.
This is part of a nominal `slogan' that \emph{atoms are interpreted by themselves}.
\end{rmrk}

\begin{defn}
\label{defn.finite.equivariant}
Suppose $\ns X$ and $\ns Y$ are sets with a permutation action.
Call a function $f$ from $|\ns X|$ to $|\ns Y|$ \deffont{equivariant} when 
\begin{equation}
\Forall{\pi\in\mathbb P}\Forall{x\in|\ns X|}\pi\act (f(x))=f(\pi\act x) .
\tag{equivariance for functions}
\end{equation}
\end{defn}

\begin{lemm}
\label{lemm.equivar.no.increase.support}
If $f$ from $|\ns X|$ to $|\ns Y|$ is equivariant then $\f{supp}(f(x))\subseteq\f{supp}(x)$ for all $x\in|\ns X|$.
\end{lemm}
\begin{proof}
If $\pi\in\f{fix}(\f{supp}(x))$ then $\pi\act f(x)=f(\pi\act x)$, and if $\pi$ is finite then $\pi\act x=x$.
\end{proof}

\begin{defn}
\label{defn.pnl.interp}
Suppose $\mathcal S=(\mathcal A,\mathcal B,\mathcal F,\mathcal P,\f{ar})$ is a signature (Definition~\ref{defn.signature}).

A \deffont{(PNL) interpretation} $\mathcal I$ for $\mathcal S$ consists of the following data:
\begin{itemize*}
\item
An interpretation for the sort-signature $(\mathcal A,\mathcal B)$ (Definition~\ref{defn.interpretation}).
\item
For every $\tf f\in\mathcal F$ with $\f{ar}(\tf f)=(\alpha)\basesort$ an equivariant function $\tf f^\iden$ from $\model{\alpha}$ to $\model{\basesort}$ (Definition~\ref{defn.finite.equivariant}).
\item
For every $\tf P\in\mathcal P$ with $\f{ar}(\tf P)=\alpha$ a finite equivariant function $\tf P^\iden$ from $\model{\alpha}$ to $\{0,1\}$ (Definition~\ref{defn.atoms.perm}).
\end{itemize*}
\end{defn}

\begin{defn}
\label{defn.valuation}
Suppose $\mathcal I$ is an interpretation for $\mathcal S$.
A \deffont{valuation} $\varsigma$ to $\mathcal I$ is a map on unknowns such that for each unknown $X$,\ 
\begin{itemize*}
\item
$\varsigma(X)\in\model{\sort(X)}$,\ and\ 
\item
$\f{supp}(\varsigma(X))\subseteq \pmss(X)$.
\end{itemize*}
$\varsigma$ will range over valuations.
\end{defn}

\subsubsection{Interpretation of terms} 

\begin{defn}
\label{defn.interpret.terms}
Suppose $\mathcal I$ is an interpretation of a signature $\mathcal S$.
Suppose $\varsigma$ is a valuation to $\mathcal I$.

Define an \deffont{interpretation} 
$\denot{\mathcal I}{\varsigma}{r}$ in $\mathcal S$ by:
\begin{frameqn}
\begin{array}{r@{\ }l@{\qquad}r@{\ }l}
\denot{\mathcal I}{\varsigma}{a} =& a
&
\denot{\mathcal I}{\varsigma}{[a]r} =& [a]\denot{\mathcal I}{\varsigma}{r}
\\
\denot{\mathcal I}{\varsigma}{\tf f(r)} =& 
\tf f^\iden(\denot{\mathcal I}{\varsigma}{r})
&
\denot{\mathcal I}{\varsigma}{\pi\act X} =& \pi\act\varsigma(X)
\\
\denot{\mathcal I}{\varsigma}{(r_1,\ldots,r_n)} =& 
(\denot{\mathcal I}{\varsigma}{r_1},\ldots,\denot{\mathcal I}{\varsigma}{r_n})
\end{array}
\end{frameqn}
\end{defn}

\begin{lemm}
\label{lemm.sort.r}
If $r:\alpha$ then $\denot{\mathcal I}{\varsigma}{r}\in\model{\alpha}$.
\end{lemm}
\begin{proof}
By a routine induction on $r$.
\end{proof} 

\begin{lemm}
\label{lemm.pi.r.model}
$\pi\act\denot{\mathcal I}{\varsigma}{r} = \denot{\mathcal I}{\varsigma}{\pi\act r}$.
\end{lemm}
\begin{proof}
By a routine induction on $r$. 
We consider one case:
\begin{itemize}
\item
The case $\pi'\act X$.\quad
By Definition~\ref{defn.interpret.terms} 
$\denot{\mathcal I}{\varsigma}{\pi'\act X} = \pi'\act \varsigma(X)$.
Therefore $\pi\act
\denot{\mathcal I}{\varsigma}{\pi'\act X} = \pi\act(\pi'\act \varsigma(X))$.
It is a fact of the group action (Definition~\ref{defn.perm.set}) that $\pi\act(\pi'\act\varsigma(X))=(\pi\circ\pi')\act\varsigma(X)$, and of the permutation action (Definition~\ref{defn.permutation.action}) that $\pi\act(\pi'\act X)\equiv (\pi\circ\pi')\act X$.
The result follows.
\qedhere
\end{itemize}
\end{proof}

\begin{lemm}
\label{lemm.supp.r}
$\f{supp}(\denot{\mathcal I}{\varsigma}{r})\subseteq\f{fa}(r)$.
\end{lemm}
\begin{proof}
By a routine induction on $r$.
We consider one case in detail:
\begin{itemize}
\item
The case $\pi\act X$.\quad
$\f{fa}(\pi\act X)=\pi\act\pmss(X)$ by Definition~\ref{defn.fa}.
By assumption in Definition~\ref{defn.valuation} $\f{supp}(\varsigma(X))\subseteq\pmss(X)$.
\end{itemize} 
The cases of $a$, $[a]r$, and $[a]r$ use parts~1, 2, and 3 of Lemma~\ref{lemm.properties.of.support}.
The case of $\tf f$ uses Lemma~\ref{lemm.equivar.no.increase.support}.
\end{proof}

\subsubsection{Interpretation of propositions} 

\begin{defn}
Suppose $\varsigma$ is a valuation to an interpretation $\mathcal I$.
Suppose $X$ is an unknown and $x\in \model{\sort(X)}$ is such that $\f{supp}(x)\subseteq \pmss(X)$.
Define a function $\varsigma[X\ssm x]$ by
$$
(\varsigma[X\ssm x])(Y)=\varsigma(Y)
\quad\text{and}\quad
(\varsigma[X\ssm x])(X)=x .
$$
\end{defn}
It is easy to verify that $\varsigma[X\ssm x]$ is also a valuation to $\mathcal I$.

\begin{defn}
\label{defn.truth}
Suppose that $\mathcal I$ is an interpretation.
Define an \deffont{interpretation of propositions} by:
\begin{frameqn}
\begin{array}{r@{\ }l}
\denot{\mathcal I}{\varsigma}{\tf P(r)} =&
\tf P^\iden(\denot{\mathcal I}{\varsigma}{r})
\\
\denot{\mathcal I}{\varsigma}{\bot}=& 0
\\
\denot{\mathcal I}{\varsigma}{\phi\limp\psi}=& 
\f{max}\{1{-}\denot{\mathcal I}{\varsigma}{\phi},\denot{\mathcal I}{\varsigma}{\psi}\}
\\
\denot{\mathcal I}{\varsigma}{\Forall{X}\phi}=&\f{min}\{\denot{\mathcal I}{\varsigma[X{\ssm} x]}{\phi}\mid x{\in} \model{\sort(X)},\, \f{supp}(x){\subseteq} \pmss(X)
\}
\end{array}
\end{frameqn}
\end{defn}

\begin{lemm}
\label{lemm.no.change}
$
\denot{\mathcal I}{\varsigma}{\phi}=\denot{\mathcal I}{\varsigma}{\pi\act\phi}$ always.
\end{lemm}
\begin{proof}
By induction on $\phi$.
We consider two cases:
\begin{itemize}
\item
The case $\Forall{X}\phi$.\quad
Suppose $\denot{\mathcal I}{\varsigma}{\Forall{X}\phi} = 1$.
This means that $\denot{\mathcal I}{\varsigma[X\ssm x]}{\phi} = 1$ 
for all $x\in\model{\alpha}$ such that $\f{supp}(x)\subseteq \pmss(X)$.
By inductive hypothesis $\denot{\mathcal I}{\varsigma[X\ssm
    x]}{\pi\act\phi} = 1$ for all $x\in\model{\alpha}$ such that 
$\f{supp}(x)\subseteq \pmss(X)$.
Therefore $\denot{\mathcal I}{\varsigma}{\Forall{X}\pi\act\phi} = 1$.
The result follows, since $\pi\act(\Forall{X}\phi)\equiv \Forall{X}\pi\act\phi$.
\item
The case $\tf P(r)$.\quad
We have 
$\denot{\mathcal I}{\varsigma}{\tf P(r)} = 
\tf P^\iden
(\denot{\mathcal I}{\varsigma}{r})$.
As $\tf P^\iden$ is equivariant, we get
$\denot{\mathcal I}{\varsigma}{\tf P(r)} = 
\tf P^\iden
(\pi\act \denot{\mathcal I}{\varsigma}{r})$.
By Lemma~\ref{lemm.pi.r.model} $\pi\act\denot{\mathcal
  I}{\varsigma}{r} = 
\denot{\mathcal I}{\varsigma}{\pi\act r}$.
Thus 
$\denot{\mathcal I}{\varsigma}{\tf P(r)} = 
\tf P^\iden(\denot{\mathcal I}{\varsigma}{\pi\act r}) = 
\denot{\mathcal I}{\varsigma}{\pi\act \tf P(r)}$. 
\qedhere
\end{itemize}
\end{proof}

\maketab{tab0}{@{\hspace{-2em}}R{8em}@{\ \ $=$\ \ }L{8em}L{10em}}

\begin{lemm}
\label{lemm.denotsub}
\begin{itemize*}
\item
$\denot{\mathcal I}{\varsigma[X\ssm \denot{\mathcal I}{\varsigma}{t}]}{r}
=\denot{\mathcal I}{\varsigma}{r[X\ssm t]}$.
\item
$\denot{\mathcal I}{\varsigma[X\ssm \denot{\mathcal I}{\varsigma}{t}]}{\phi} =
\denot{\mathcal I}{\varsigma}{\phi[X\ssm t]}$.
\end{itemize*}
\end{lemm}
\begin{proof}
By routine inductions on the definitions of 
$\denot{\mathcal I}{\varsigma}{r}$ and 
$\denot{\mathcal I}{\varsigma}{\phi}$ in 
Definitions~\ref{defn.interpret.terms} and~\ref{defn.truth}.
We consider two cases:
\begin{itemize*}
\item
The case of $\denot{\mathcal I}{\varsigma[X\ssm t]}{\pi\act X}$.\quad
We reason as follows:
\begin{tab0}
\denot{\mathcal I}{\varsigma[X\ssm \denot{\mathcal I}{\varsigma}{t}]}{\pi\act X}
& \pi\act \denot{\mathcal I}{\varsigma}{t} 
&\text{Definition~\ref{defn.interpret.terms}}
\\
& \denot{\mathcal I}{\varsigma}{\pi\act t}
&\text{Lemma~\ref{lemm.pi.r.model}}
\\
& \denot{\mathcal I}{\varsigma}{(\pi\act X)[X\ssm t]}
&\text{Definition~\ref{defn.subst.action}} .
\end{tab0}
\item
The case of $\denot{\mathcal I}{\varsigma[X\ssm t]}{\tf P(r)}$.
\quad
We reason as follows:
\begin{tab0}
\denot{\mathcal I}{\varsigma[X\ssm \denot{\mathcal I}{\varsigma}{t}]}{\tf P(r)}
&
{\tf P}^\iden(\denot{\mathcal I}{\varsigma[X\ssm \denot{\mathcal I}{\varsigma}{t}]}{r})
&\text{Definition~\ref{defn.truth}}
\\
&
{\tf P}^\iden(\denot{\mathcal I}{\varsigma}{r[X\ssm t]}) 
&\text{Part~1 of this result}
\\
&
\denot{\mathcal I}{\varsigma}{\tf P(r)[X\ssm t]} 
&\text{Definition~\ref{defn.truth}} .
\end{tab0}
\end{itemize*}
\end{proof}

\begin{lemm}
\label{lemm.fV.denot}
If $\varsigma(X)=\varsigma'(X)$ for all $X\in\f{fV}(r)$ then $\denot{\mathcal I}{\varsigma}{r}=\denot{\mathcal I}{\varsigma'}{r}$, and similarly for $\phi$.
\end{lemm}
\begin{proof}
By a routine induction on $r$ and $\phi$.
\end{proof}

\subsubsection{Validity and soundness} 

\begin{defn}[Validity]
\label{defn.pnl.ment}
Call the proposition $\phi$ \deffont{valid} in ${\mathcal I}$ when 
$\denot{\mathcal I}{\varsigma}{\phi} = 1$ for all $\varsigma$. 

Call the sequent $\phi_1, ..., \phi_n \cent \psi_1, ..., \psi_p$ \deffont{valid} 
in ${\mathcal I}$ when 
$(\phi_1 \wedge ... \wedge \phi_n) \Rightarrow 
(\psi_1 \vee ... \vee \psi_p)$ is valid. 
\end{defn}

\begin{thrm}[Soundness]
\label{thrm.soundness}
If $\Phi\cent\Psi$ is derivable, then it is valid in all interpretations.
\end{thrm} 
\begin{proof}
By induction on derivations (Figure~\ref{Seq}). 
The case of \rulefont{Ax} uses Lemma~\ref{lemm.no.change}.
The case of \rulefont{\forall L} uses Lemma~\ref{lemm.denotsub}.
The case of \rulefont{\forall R} uses Lemma~\ref{lemm.fV.denot}.
Other rules are routine by unpacking definitions.
\end{proof}

\subsection{Completeness}
\label{subsect.completeness}

In this subsection we prove Theorem~\ref{thrm.completeness}: a converse to Theorem~\ref{thrm.soundness}, that if $\phi$ is valid in all models, then $\phi$ it is derivable.

For this subsection fix the following data:
\begin{itemize*}
\item
A signature $\mathcal S=(\mathcal A,\mathcal B,\mathcal F,\mathcal P,\f{ar})$.
\item
A formula $\phi$ such that $\not\cent\phi$.
\end{itemize*}
We will construct an interpretation $\mathcal I$ and a valuation $\varsigma$ (Definition~\ref{defn.interpretation}) such that $\denot{\mathcal I}{\varsigma}{\phi} = 0$.
This suffices to prove the result.

\subsubsection{Maximally consistent set of propositions}

\begin{defn}
\label{defn.herbrand}
Choose a fixed but arbitrary enumeration of propositions $\xi_1$, $\xi_2$, $\xi_3$, \ldots

Define $\Phi_1=\{\neg\phi\}$.
For each $i\geq 1$ we define $\Phi_{i+1}$ as follows:
\begin{itemize*}
\item
If $\Phi_i\cent\xi_i$ then write $\xi=\xi_i$.
\item
If $\Phi_i\cent\neg\xi_i$ then write $\xi=\neg\xi_i$.
\item
If $\Phi_i\not\cent\xi_i$ and $\Phi_i\not\cent\neg\xi_i$ then write $\xi=\xi_i$.
\end{itemize*}
There are now two cases:
\begin{itemize*}
\item
If $\xi$ has the form $\neg\Forall{X}\xi'$ then we define $\Phi_{i+1}=\Phi_i\cup\{\xi,\neg\xi'[X\ssm Z]\}$ where $Z$ is some fixed but arbitrary choice of unknown that is not free in any proposition in $\Phi_i$ and is such that $\pmss(Z)=\pmss(X)$ and $\sort(Z)=\sort(X)$.
\item
Otherwise, we define $\Phi_{i+1}=\Phi_i\cup\{\xi\}$.
\end{itemize*} 
Finally, we define $\Phi_\omega$ by
$\Phi_\omega=\bigcup_i \Phi_i$.
\end{defn}

\begin{lemm}
\label{lemm.2}
For every $i$,\ \ $\Phi_i\not\cent\bot$.
\end{lemm}
\begin{proof}
By induction on $i$:
\begin{itemize*}
\item
By definition $\Phi_1=\{\neg\phi\}$.
As $\not\cent\phi$, we have $\neg \phi \not\cent \bot$ 
\item
Suppose $\Phi_i\not\cent\bot$. 
Either $\Phi_{i+1}=\Phi_i\cup\{\neg\xi\}$ or
$\Phi_{i+1}=\Phi_i\cup\{\neg\xi,\neg\xi'[X\ssm Z]\}$---we consider
the first, simpler case; the second case is similar.  Suppose $\Phi_i,
\xi \cent \bot$. It follows that $\Phi_i \cent \neg \xi$. So we are
not in the third case of Definition \ref{defn.herbrand} and we are
either in the first or the second. So $\Phi_i \cent \xi$ and thus
$\Phi_i \cent \bot$; a contradiction.
\qedhere\end{itemize*}
\end{proof}

\begin{lemm}
\label{lemm.3}
$\Phi_{\omega} \not\cent \bot$.
\end{lemm}
\begin{proof}
Assume $\Phi_{\omega} \cent \bot$.
So there exists a finite subset  $\Gamma$ of $\Phi_{\omega}$ such that $\Gamma \cent \bot$. 
As $\Gamma$ is finite it is included in some $\Phi_i$, and $\Phi_i \cent \bot$, contradicting Proposition~\ref{lemm.2}.
\end{proof}

\begin{lemm}
\label{lemm.1}
For every $\xi$, at least one of $\xi\in \Phi_\omega$ and  $\neg\xi\in \Phi_\omega$ holds.
\end{lemm}
\begin{proof}
We check of Definition~\ref{defn.herbrand} that
for every $i$, either $\xi_i\in\Phi_{i+1}$ or $\neg\xi_i\in\Phi_{i+1}$.
The result follows.
\end{proof}

\begin{lemm}
\label{lemm.1bis}
For every $\xi$, if\ $\neg \Forall{X}\xi \in \Phi_\omega$ then there exists  a $Z$ such that $\neg \xi[X\ssm Z] \in \Phi_\omega$. 
\end{lemm}
\begin{proof}
There exists an $i$ such that 
$\xi_i = \neg \Forall{X}\xi$.
Since $\Phi_{\omega} \cent \xi_i$ and $\Phi_{\omega} \not\cent \bot$, we have that 
$\Phi_{\omega} \not\cent \neg \xi_i$, and so $\Phi_{i} \not\cent \neg \xi_i$. 
Thus $\Phi_{i+1} = \Phi_i \cup \{\neg \Forall{X}\xi,\ \neg \xi[X\ssm Z]\}$.
The result follows.
\end{proof}

\begin{lemm}
\label{lemm.deductively.closed}
If $\Phi_{\omega}\cent\phi$ then $\phi \in\Phi_\omega$.
\end{lemm}
\begin{proof}
As, by Lemma~\ref{lemm.3}, $\Phi_{\omega}\not\cent\bot$, 
if $\Phi_{\omega}\cent\phi$ then $\neg \phi \not\in \Phi_{\omega}$. Thus 
by Lemma~\ref{lemm.1}, $\phi \in \Phi_{\omega}$.
\end{proof}

\begin{corr}
\label{corr.char.inclusion}
\begin{enumerate}
\item
$(\phi\limp\psi)\in\Phi_\omega$ if and only if ($\phi\not\in\Phi_\omega$ or $\psi\in\Phi_\omega$).
\item
$\Forall{X}\phi\in\Phi_\omega$ if and only if 
\\
\qquad (for every $r$ such that $r:\sort(X)$ and $\f{fa}(r)\subseteq\pmss(X)$,\ $\phi[X\ssm r]\in\Phi_\omega$).
\end{enumerate}
\end{corr}
\begin{proof}
\begin{enumerate}
\item
Suppose $(\phi \Rightarrow \psi) \in \Phi_{\omega}$ and
$\phi \in \Phi_{\omega}$.
Then $\Phi_{\omega} \cent \psi$ and so by Lemma~\ref{lemm.deductively.closed} $\psi \in \Phi_{\omega}$.

Suppose $\phi \not\in \Phi_{\omega}$.
By Lemma~\ref{lemm.1} $\neg \phi \in  \Phi_{\omega}$.
So $\Phi_{\omega} \cent \neg \phi$  and therefore $\Phi_{\omega} \cent \phi \Rightarrow \psi$. 
By Lemma~\ref{lemm.deductively.closed}  $(\phi \Rightarrow \psi) \in \Phi_{\omega}$.

Suppose $\psi \in \Phi_{\omega}$. 
Then $\Phi_{\omega} \cent \psi$ and so $\Phi_{\omega} \cent \phi \Rightarrow \psi$.
By Lemma~\ref{lemm.deductively.closed} $(\phi \Rightarrow \psi) \in \Phi_{\omega}$.
\item
Suppose $\Forall{X}\phi \in \Phi_{\omega}$.
By Lemma~\ref{lemm.deductively.closed}, if $r : \sort(X)$ and $\fa(r) \subseteq \pmss(X)$ then 
$\phi[X\ssm r] \in \Phi_{\omega}$.

Conversely, suppose $\phi[X\ssm r] \in \Phi_{\omega}$ for every $r$ such that $r : \sort(X)$ and $\fa(r) \subseteq \pmss(X)$.
We proceed by contradiction: suppose $\Forall{X}\phi \not\in \Phi_{\omega}$.
By Lemma~\ref{lemm.1}  $\neg \Forall{X}\phi \in \Phi_{\omega}$ and 
by Lemma~\ref{lemm.1bis}, there exists a $Z$ such that 
$\neg \phi[X\ssm Z] \in \Phi_{\omega}$. 
So $\Phi_{\omega} \cent \neg \phi[X\ssm Z]$, and so 
$\Phi_{\omega} \cent \phi[X\ssm Z]$, and so $\Phi_{\omega} \cent \bot$, 
contradicting Lemma~\ref{lemm.3}.
\qedhere\end{enumerate}
\end{proof}

\subsubsection{The term model}

\begin{defn}
\label{defn.term.model}
Define $\mathcal I$ by:
\begin{itemize*}
\item
$\model{\alpha}=\{r\mid r:\alpha\}$.
\item
$\tf f^\iden$ maps $r$ to $\tf f(r)$.
\item
$\tf P^\iden$ maps $r_1,\ldots,r_n$ to $1$ if 
$\tf P(r_1,\ldots,r_n)\in\Phi_\omega$ and to $0$ otherwise.
\end{itemize*}
Define $\varsigma$ by
$$
\varsigma(X)=X.
$$
We endow $\model{\alpha}$ with a permutation action given by the action on terms.
\end{defn}

\begin{lemm}
\label{lemm.supp.raeq}
$\f{supp}(r)=\f{fa}(r)$.

As a corollary, $\model{\alpha}$ from Definition~\ref{defn.term.model} is a permissive-nominal set.
\end{lemm}
\begin{proof}
We prove two subset inclusions:
\begin{itemize}
\item
\emph{Proof that ${\f{supp}(r)\subseteq\f{fa}(r)}$.}\quad
By Lemma~\ref{lemm.fa.fix}, if $\pi(a)=a$ for all $a\in\f{fa}(r)$ then $r= \pi\act r$.
It follows by the definition of support (Theorem~\ref{thrm.supp}) that $\f{supp}(r)\subseteq\f{fa}(r)$.
\item
\emph{Proof that ${\f{fa}(r)\subseteq\f{supp}(r)}$.}\quad
Suppose $a\in\f{fa}(r)$.
Choose some fresh $b$ (so $b\not\in\f{fa}(r)\cup\f{supp}(r)$).
By Lemma~\ref{lemm.fa.pi.r} $\f{fa}((b\ a)\act r)=(b\ a)\act\f{fa}(r)$.
Since $(b\ a)\act\f{fa}(r)\neq\f{fa}(r)$ it follows using Lemma~\ref{lemm.fa.aeq} that $r\neq (b\ a)\act r$.
The result follows by the first part of this result and by Corollary~\ref{corr.notinsupp}.
\end{itemize}
The corollary is immediate, unpacking Definition~\ref{defn.support}.
\end{proof}

\begin{lemm}
\label{lemm.f.fin.equivar}
\begin{enumerate}
\item
If $\f{ar}(\tf f)=(\alpha)\basesort$ then $\tf f^\iden$ is
 well-defined, equivariant, and maps $\model{\alpha}$ to 
$\model{\basesort}$.
\item
If $\f{ar}(\tf P)=\alpha$ then $\tf P^\iden$ is well-defined,
equivariant, and maps $\model{\alpha}$ to $\{0,1\}$.
\end{enumerate}
\end{lemm}
\begin{proof}
\begin{enumerate}
\item
The only (very) slightly non-trivial part is equivariance.
We reason as follows:
\begin{tab0}
\pi\act {\tf f}^\iden(r)&\pi\act \tf f(r)
&\text{Definition~\ref{defn.interpret.terms}}
\\
&\tf f(\pi\act r)
&\text{Definition~\ref{defn.permutation.action}}
\\
&{\tf f}^\iden(\pi\act r) 
&\text{Definition~\ref{defn.interpret.terms}}
\end{tab0}
\item
The case of $\tf P$ is similar.
\qedhere\end{enumerate}
\end{proof}

\begin{prop}
$\mathcal I$ is an interpretation of the signature $\mathcal S=(\mathcal A,\mathcal B,\mathcal F,\mathcal P,\f{ar})$ which we fixed at the beginning of this subsection.
In addition, $\varsigma$ is a valuation to $\mathcal I$.
\end{prop}
\begin{proof}
By Lemma~\ref{lemm.supp.raeq} each $\model{\alpha}$ is a permissive-nominal set.
By Lemma~\ref{lemm.f.fin.equivar} for each $\tf f:(\alpha)\basesort\in\mathcal F$,\ ${\tf f}^\iden$ is an equivariant map
from $\model{\alpha}$ to $\model{\basesort}$ and for each $\tf
P:\alpha\in\mathcal P$,\ ${\tf P}^\iden$ is a finite equivariant function from $\model{\alpha}$ to $\{0,1\}$.

By construction $\varsigma(X)\in\model{\sort(X)}$ always.
By Lemma~\ref{lemm.supp.raeq} $\f{supp}(\varsigma(X))=\f{supp}(X)\subseteq\pmss(X)$ always.

The result follows.
\end{proof}

\begin{lemm}
\label{lemm.rvarsigma}
$\denot{\mathcal I}{\varsigma}{r} = r$.
\end{lemm}
\begin{proof}
By a routine induction on the definition of $\denot{\mathcal I}{\varsigma}{r}$  in Definition~\ref{defn.interpret.terms}.
We consider just one case:
\begin{itemize*}
\item
The case of $\denot{\mathcal I}{\varsigma}{\pi\act X}$.
We reason as follows:
\begin{tab0}
\denot{\mathcal I}{\varsigma}{\pi\act X} & \pi\act \varsigma(X)
&\text{Definition~\ref{defn.interpret.terms}}
\\
& \pi\act X 
&\text{Definition~\ref{defn.term.model}} .
\end{tab0}
\end{itemize*}
\end{proof}

\begin{lemm}
\label{lemm.varsigma.in.xi}
$\denot{\mathcal I}{\varsigma}{\xi} = 1$ if and only if $\xi\in\Phi_\omega$.
\end{lemm}
\begin{proof}
By induction on the definition of $\denot{\mathcal I}{\varsigma}{\xi}$ 
(Definition~\ref{defn.truth}):
\begin{itemize*}
\item
The case of $\denot{\mathcal I}{\varsigma}{\tf P(r)}$.\quad
We reason as follows:
\maketab{tab1}{@{\hspace{-5em}}R{8em}@{\ \ $\liff$\ \ }L{18em}L{10em}}
\begin{tab1}
\denot{\mathcal I}{\varsigma}{\tf P(r)} = 1& 
{\tf P}^\iden(\denot{\mathcal I}{\varsigma}{r}) = 1
&\text{Definition~\ref{defn.truth}}
\\
 & {\tf P}^\iden(r) = 1
&\text{Lemma~\ref{lemm.rvarsigma}}
\\
 & \tf P(r)\in \Phi_\omega
&\text{Definition~\ref{defn.term.model}}
\end{tab1}
\item
The case of $\denot{\mathcal I}{\varsigma}{\bot}$.\quad 
By definition $\denot{\mathcal I}{\varsigma}{\bot}=0$.
By part~1 of Corollary~\ref{corr.char.inclusion},\ $\bot\not\in\Phi_\omega$.
\item
The case of $\denot{\mathcal I}{\varsigma}{\phi\limp\psi}$.\quad
We reason as follows:
\begin{tab1}
\denot{\mathcal I}{\varsigma}{\phi\limp\psi} = 1& 
\denot{\mathcal I}{\varsigma}{\phi} = 0~\mbox{or}~
\denot{\mathcal I}{\varsigma}{\psi} = 1
&\text{Definition~\ref{defn.truth}}
\\
 & \phi\not\in\Phi_\omega~\mbox{or}~\psi\in\Phi_\omega
&\text{ind. hyp.}
\\
 & \phi\limp\psi\in\Phi_\omega 
&\text{Corollary~\ref{corr.char.inclusion}, part~2}
\end{tab1}
\item
The case of $\denot{\mathcal I}{\varsigma}{\Forall{X}\phi}$, 
where $\alpha=\sort(X)$ and $S=\pmss(X)$.\quad
\begin{tab1}
\denot{\mathcal I}{\varsigma}{\Forall{X}\phi} = 1 &
\Forall{t\in\model{\alpha}}\f{supp}(t)\subseteq S\limp 
\denot{\mathcal I}{\varsigma[X\ssm t]}{\phi} = 1
&\text{Definition~\ref{defn.truth}}
\\
&
\Forall{t\in\model{\alpha}}\f{supp}(t)\subseteq S\limp 
\denot{\mathcal I}{\varsigma}{\phi[X\ssm t]} = 1
&\text{Lems.~\ref{lemm.fV.denot}, \ref{lemm.rvarsigma}}
\\
&
\denot{\mathcal I}{\varsigma}{\phi[X\ssm t]} = 1\text{ every }t:\alpha\text{ s.t. }\f{fa}(t)\subseteq S 
&\text{Lem.~\ref{lemm.supp.raeq}}
\\
&
\phi[X\ssm t]\in\Phi_\omega \text{ every }t:\alpha\text{ s.t. }\f{fa}(t)\subseteq S 
&\text{ind. hyp.}
\\
&
\Forall{X}\phi\in\Phi_\omega 
&\text{Cor.~\ref{corr.char.inclusion}, part~4}
\end{tab1}
\end{itemize*}
\end{proof}

\begin{lemm}
\label{lemm.varsigma.in.xi.not}
If $\not\cent\phi$, then there exists a model $\mathcal I$ and a
valuation $\varsigma$ such that $\denot{\mathcal I}{\varsigma}{\phi} =
0$. 
\end{lemm}
\begin{proof}
As $\neg \phi {\in} \Phi_0 {\subseteq} \Phi_{\omega}$ and 
$\Phi_{\omega} {\not\cent} \bot$, we have $\phi {\not\in} \Phi_{\omega}$. 
By Lemma~\ref{lemm.varsigma.in.xi}, we get 
$\denot{\mathcal I}{\varsigma}{\phi} = 0$. 
\end{proof}

As a corollary we get Theorem~\ref{thrm.completeness}:
\begin{frametxt}
\begin{thrm}[Completeness]
\label{thrm.completeness}
If $\phi$ is valid in all models, then $\phi$ is derivable.
\end{thrm}
\end{frametxt}

\section{Specifying arithmetic in permissive-nominal logic}
\label{sect.pnl.arithmetic}

We start by defining the sorts, term-formers, and proposition-formers of a signature $\dot{\mathcal L}$ which is suitable for finitely specifying arithmetic in PNL.
We then specify its axioms and, in Subsection~\ref{subsect.comments.on.the.axioms}, we discuss them in detail.

\subsection{The signature $\dot{\mathcal L}$ and the axioms}
\label{subsect.dotl}

\begin{defn}
\label{defn.fotl}
A signature $\dot{\mathcal L}$ suitable for writing out a PNL theory of first-order logic is given in Figure~\ref{fig.dotl}.
\end{defn}

\begin{figure}
We assume one atomic sort $\nu$ and two base sorts $\iota$ and $o$.

We assume term-formers and proposition-formers as follows: 
$$
\begin{array}{c@{\qquad}c@{\qquad}c@{\qquad}c@{\qquad}c}
\dotzero:\iota & 
\dotsucc:(\iota)\iota & 
\dotplus:(\iota,\iota)\iota & 
\dottimes:(\iota,\iota)\iota &
\\
\dotbot:o &
\dotlimp:(o,o)o &
\dotforall:([\nu]o)o &
\dotoeq:(\iota,\iota)o &
\\
\tf{var}:(\nu)\iota
&
\tf{sub}_\iota:([\nu]\iota,\iota)\iota 
&
\tf{sub}_o:([\nu]o,\iota)o
\\[2ex]
\pnleq_\iota:(\iota,\iota)&  \pnleq_o:(o,o)& \epsilon:(o) 
\end{array}
$$ 
\caption{Signature suitable for a PNL specification of arithmetic}
\label{fig.dotl}
\end{figure}

We introduce the following syntactic sugar:
\begin{itemize*}
\item
We may write $\tf{sub}_o([a]r,t)$ as $r[a\sm t]$.
\item
We may write $\tf{sub}_\iota([a]r,t)$ as $r[a\sm t]$.
\item
We may write both $\pnleq_\iota$ and $\pnleq_o$ just as $\pnleq$.
\end{itemize*}
Examples of this syntactic sugar in use, follow immediately below.

\subsubsection*{Equality}

Axioms for equality $\pnleq:(\iota,\iota)$ and equality $\pnleq:(o,o)$ are given in Figure~\ref{fig.equality}.

\subsubsection*{Substitution}

Axioms for substitution $\tf{sub}_\iota$ and $\tf{sub}_o$ are given in Figure~\ref{fig.substitution}. 

We arguably abuse notation in Figure~\ref{fig.substitution} when we use variables of sort $\iota$ and $o$ as appropriate not necessarily giving them distinct names (e.g. in \rulefont{sub{*}} $X$ has sort $\iota$, whereas in \rulefont{sub{\dotlimp}} we use another variable also written $X$ with sort $o$).

\begin{figure}
$$
\begin{array}{l@{\quad}r@{}r@{\ }c@{\ }l@{\ \;}l}
\rulefont{{\pnleq}2}&
\Forall{X',X,Y',Y} & (X'{\pnleq}X{\land} Y'{\pnleq} Y) &\limp& X' \mathrel{op} Y'\pnleq X\mathrel{op} Y
&\f{op}{\in}\{\dotplus,\dottimes,\dotlimp,\dotoeq\}
\\
\rulefont{{\pnleq}1}&
\Forall{X',X}&X'{\pnleq} X&\limp& \f{op}(X')\pnleq \f{op}(X)
&\f{op}{\in}\{\dotsucc\}
\\
\rulefont{{\pnleq}0}&
\Forall{X}&&&X\pnleq X 
\\
\rulefont{{\pnleq}\dotforall}&
\Forall{Z',Z}&Z'{\pnleq} Z&\limp& \dotforall ([a]Z')\pnleq \dotforall ([a]Z) 
\\
\rulefont{{\pnleq}\tf{sub}}&
\Forall{X',X,Y',Y}& (X'{\pnleq} X{\land}Y'{\pnleq} Y) &\limp& \f{op}([a]X',Y') \pnleq \f{op}([a]X,Y)
&\f{op}{\in}\{\tf{sub}_\iota,\tf{sub}_o\}
\\
\rulefont{{\pnleq}\mathit o}&
\Forall{Z',Z}&Z'{\pnleq} Z&\limp& (\epsilon(Z')\liff \epsilon(Z))
\\
\rulefont{{\pnleq}\iota}&
\Forall{X',X}& X'{\pnleq} X &\limp & \epsilon(X'\dotoeq X) 
\end{array}
\vspace{-1ex}$$
We fill in sorts as appropriate.
Thus, $\dotbot\pnleq_o \dotbot$ whereas $0\pnleq_\iota 0$, and so on.
The permission sets of all variables are equal to $\atomsdown$, and $a\in\atomsdown$.
\caption{\theory{EQU}: axioms for equality as a PNL theory}
\label{fig.equality}
\end{figure}
\begin{figure}
$$
\begin{array}{lr@{\ }l@{\ }c@{\ }ll}
\rulefont{sub\tf{var}}&\Forall{X}&\tf{var}(a)[a\sm X] &\pnleq& X 
\\
\rulefont{sub\#}&\Forall{X,Z}&Z[a\sm X] &\pnleq& Z 
&
\hspace{-4em}(\pmss(Z)=(b\ a)\act\atomsdown )
\\
\rulefont{sub\dotsucc}& \Forall{X',X}&\dotsucc(X')[a\sm X]&\pnleq& \dotsucc(X'[a\sm X])
\\
\rulefont{sub{\dotplus}}& \Forall{X'',X',X}&(X''\dotplus X')[a\sm X]&\pnleq& (X''[a\sm X]\dotplus X'[a\sm X])
\\
\rulefont{sub{\dottimes}}& \Forall{X'',X',X}&(X''\dottimes X')[a\sm X]&\pnleq& (X''[a\sm X]\dottimes X'[a\sm X])
\\
\rulefont{sub{\dotlimp}}& \Forall{X'',X',X}&(X''\dotlimp X')[a\sm X]&\pnleq& (X''[a\sm X]\dotlimp X'[a\sm X])
\\
\rulefont{sub{\dotoeq}}& \Forall{X'',X',X}&(X''\dotoeq X')[a\sm X]&\pnleq& (X''[a\sm X]\dotoeq X'[a\sm X])
\\
\rulefont{sub\dotforall}&\Forall{X,Z}& (\dotforall([b]Z))[a\sm X] &\pnleq&  \dotforall([b](Z[a\sm X])) 
&
\hspace{-4em}(\pmss(Z)=(b\ a)\act\atomsdown )
\\
\rulefont{sub id}&\Forall{X}& X[a\sm \tf{var}(a)] &\pnleq& X 
\end{array}
\vspace{-1ex}$$
$a\in\atomsdown $ and $b\not\in\atomsdown$.
The permission set of $X''$, $X'$, and $X$ is equal to $\atomsdown$.
The permission set of $Z$ is equal to $(b\ a)\act\atomsdown$ (Definition~\ref{defn.pointwise}).
\caption{\theory{SUB}: axioms for substitution as a PNL theory}
\label{fig.substitution}
\end{figure}
\begin{figure}
$$
\begin{array}{l@{\hspace{6em}}r@{\ }l@{\ }c@{\ }ll}
\rulefont{\dotlimp}& \Forall{Z',Z} &\epsilon(Z'\dotlimp Z) &\liff& (\epsilon(Z')\limp\epsilon(Z)) &\hspace{6em}
\\
\rulefont{\dotforall} & \Forall{Z}&\bigl(\epsilon(\dotforall([a]Z)) &\liff& \Forall{X}\epsilon(Z[a\sm X])\bigr)
\\
\rulefont{\dotbot} & &\epsilon(\dotbot)&\limp&\bot
\end{array}
$$
Here $Z'$ and $Z$ have sort $o$, permission set $\atomsdown $, and $a\in\atomsdown $.
\caption{\theory{FOL}: axioms for first-order formulas as a PNL theory}
\label{fig.fol}
\end{figure}
\begin{figure}
$$
\begin{array}{l@{\hspace{6em}}r@{\ }l}
\rulefont{PS0}
&\Forall{X}&\dotsucc(X)\pnleq \dotzero \limp \bot 
\\
\rulefont{PSS} 
&\Forall{X',X}&\dotsucc(X')\pnleq \dotsucc(X) \limp  X'\pnleq X
\\
\rulefont{P{+}0}
&\Forall{X}&X\dotplus \dotzero\pnleq X
\\
\rulefont{P{+}succ}
&\Forall{X',X}&X'\dotplus\dotsucc(X)\pnleq \dotsucc(X')\dotplus X
\\
\rulefont{P{*}0}
&\Forall{X}&X\dottimes \dotzero\pnleq\dotzero 
\\
\rulefont{P{*}succ}
&\Forall{X',X}&X'\dottimes \dotsucc(X)\pnleq (X'\dottimes X)\dotplus X
\\[1ex]
\rulefont{PInd}
&\Forall{Z}&(\epsilon(Z[a\sm \dotzero])
\limp 
\\
&&\ \ \bigl(\Forall{X}(\epsilon(Z[a\sm X])\limp\epsilon(Z[a\sm \dotsucc(X)]))\bigr)\limp 
\\
&&\ \ \ \ \Forall{X}\epsilon(Z[a\sm X]) )
\end{array}
$$
All variables have permission set $\atomsdown $, and $a\in\atomsdown $.
\caption{\theory{ARITH}: axioms for arithmetic as a PNL theory}
\label{fig.arithmetic}
\end{figure}

\subsubsection*{First-order logic}

Axioms reflecting first-order formulas (in a shallow sense) as terms in PNL (the $\dotbot$, $\dotlimp$, and $\dotforall$) are given in Figure~\ref{fig.fol}.

\subsubsection*{Arithmetic}

Given \theory{EQU}, \theory{SUB}, and \theory{FOL}, it is not hard to write axioms for arithmetic in PNL.
This is in Figure~\ref{fig.arithmetic}.
Later on in Theorem~\ref{thrm.arithmetic} we prove that this \emph{is} an axiomatisation of arithmetic in PNL.

\subsection{Comments on the axioms}
\label{subsect.comments.on.the.axioms}

\begin{rmrk}
In \cite{gabbay:capasn-jv} capture-avoiding substitution is equationally axiomatised using nominal algebra in the style of \theory{SUB}.
Soundness and completeness are proved, so providing some formal sense in which the axioms of \theory{SUB} are `right'.

In \cite{gabbay:oneaah-jv} first-order logic is equationally axiomatised using nominal algebra (so the axioms involve only equality).
The axioms of \theory{FOL} are \emph{not} based on those of \cite{gabbay:oneaah-jv}.
In \theory{FOL}, we take advantage of the stronger language provided by PNL; Because PNL is already a first-order logic, we can use $\bot$, $\limp$, and $\forall$ directly to capture the behaviour of $\dotbot$, $\dotlimp$, and $\dotforall$.
In \cite{gabbay:oneaah-jv} we had to work a little harder because the ambient logic, nominal algebra, was purely equational.
\end{rmrk}

\begin{rmrk}
Instead of the axioms for equality \theory{EQU}, we could directly extend PNL by adding derivation rules Figure~\ref{Seq} as follows:
$$ 
\begin{prooftree}
\begin{aligned}
\Phi,\, r\pnleq s,\, &\phi[X\ssm r],\, \phi[X\ssm s]\cent \Psi
\\
&(\f{fa}(r)\cup\f{fa}(s)\subseteq\pmss(X))
\end{aligned}
\justifies
\Phi,\, r\pnleq s,\, \phi[X\ssm r]\cent \Psi
\using\rulefont{{\pnleq}S}
\end{prooftree}
\qquad\qquad
\begin{prooftree}
\Phi,\, r\pnleq r \cent \Psi
\justifies
\Phi\cent \Psi
\using\rulefont{{\pnleq}R}
\end{prooftree}
$$
\end{rmrk}

\begin{rmrk}
Every unknown has a sort, and a permission set.

Different choices of permission set may yield logically equivalent results.
For example, in \rulefont{sub\tf{lam}} it is not vital that $\pmss(Z)$ is \emph{exactly} $(b\ a)\act\atomsdown $.
The important point is that $a\not\in\pmss(Z)$.

Similarly, in \rulefont{sub\tf{app}} it is not vital that $\pmss(X'')=\pmss(X')$; when we use the axiom we can instantiate $X''$ and $X'$ to $r''$ and $r'$ such that $\f{fa}(r'')\neq\f{fa}(r')$, and conversely if we take $\pmss(X'')\neq\pmss(X')$ then we can still instantiate $X''$ and $X'$ to $r''$ and $r'$ such that $\f{fa}(r'')=\f{fa}(r')\subseteq\pmss(X'')\cap\pmss(X')$.
More on this in Section~\ref{sect.further.remarks}.
\end{rmrk}

\section{A theory of arithmetic in first-order logic}

We now recall first-order logic $\mathcal L$ and write Peano arithmetic in $\mathcal L$.
Our two main theorems are:
\begin{itemize*}
\item
Theorem~\ref{thrm.correctness} which states that the PNL theory of first-order logic written in $\dot{\mathcal L}$---in symbols this is $\theory{EQU}\cup\theory{SUB}\cup\theory{FOL}$---soundly interprets first-order logic $\mathcal L$; and
\item
Theorem~\ref{thrm.arithmetic} which states that the PNL theory of arithmetic \emph{in} the PNL theory of first-order logic, soundly and completely interprets ordinary Peano arithmetic in written in $\mathcal L$.\footnote{Missing from this is a proof that $\theory{EQU}\cup\theory{SUB}\cup\theory{FOL}$ soundly and \emph{completely} interprets first-order logic.
We believe this to be true and the proof should be an elementary simplification of the more involved case for arithmetic---but it is not worth writing out here. 

Arithmetic is (in first-order logic) axiomatised using a \emph{schema}.  We use PNL $\forall X$ to express them finitely; see e.g. the $\forall X$ in \rulefont{PInd} of Figure~\ref{fig.arithmetic} which models the `every $\xi$' in \rulefont{pind} in Figure~\ref{fig.fol.arithmetic}.

Finite first-order logic theories (including the empty theory) are unproblematic.  
In these proofs, soundness and completeness are only a means to the end of demonstrating how we can axiomatise finitely in a nominal first-order logic PNL, structures that without names and binding would require infinite axiom schemes or higher orders. 
}
\end{itemize*}

\subsection{First-order logic $\mathcal L$}

We will use the atoms $\mathbb A_\nu$ from $\dot{\mathcal L}$ in Section~\ref{sect.pnl.arithmetic} as variables of our first-order logic (this is not necessary, but it is convenient).
So for this section, $a,b,c,\ldots$ will range over distinct atoms in $\mathbb A_\nu$.

\begin{defn}
\label{defn.L}
Define \deffont{terms} and \deffont{formulas} of $\mathcal L$ by:
\begin{frameqn}
\begin{array}{r@{\ }l}
t ::=& a \mid 0 \mid \f{succ}(t) \mid t+t \mid t*t
\\
\xi ::=& t\oeq t \mid \bot \mid \xi\limp\xi \mid \Forall{a}\xi 
\end{array} 
\end{frameqn}
Substitution $t'[a\ssm t]$ and $\xi[a\ssm t]$ is as usual for first-order logic.
We write sequents $\Xi\cent\Chi$ where $\Xi$ and $\Chi$ are sets of formulas.
Derivability is as usual for first-order logic.
\end{defn}

\begin{defn}
\label{defn.amod}
Define a mapping $\amod{\text{-}}$ from terms and formulas of $\mathcal L$ 
to terms of $\dot{\mathcal L}$ 
by:
\begin{frameqn}
\begin{array}{r@{\ }l@{\qquad}r@{\ }l}
\amod{a}=&a
&
\amod{0}=&\dotzero
\\
\amod{\f{succ}(t)}=&\dotsucc(\amod{t})
&
\amod{t'+t}=&\amod{t'}\dotplus\amod{t}
\\
\amod{t'*t}=&\amod{t'}\dottimes\amod{t}
\\
\amod{t'\oeq t}=&\amod{t'}\dotoeq\amod{t}
&
\amod{\bot}=&\dotbot
\\
\amod{\xi'\limp\xi}=&\amod{\xi'}\mathrel{\dotlimp}\amod{\xi}
&
\amod{\forall a.\xi}=&\dotforall [a]\amod{\xi}
\end{array}
\end{frameqn}
\end{defn}

\begin{defn}
Extend $\amod{\text{-}}$ to first-order logic sequents $\Xi\cent \Chi$ as follows:
\begin{frameqn}
\amod{\Xi\cent\Chi} = \epsilon(\dotforall [a_1]\ldots\dotforall [a_n]\amod{(\xi_1\land \ldots\land \xi_k) \limp (\chi_1\lor\ldots\lor\chi_l)})
\end{frameqn}
Here, $\Xi=\{\xi_1,\ldots,\xi_k\}$,\ $\Chi=\{\chi_1,\ldots,\chi_l\}$,\ and the free variables of $\Xi$ and $\Chi$ are $\{a_1,\ldots,a_n\}$ (in some order).
\end{defn}

\begin{nttn}
\label{nttn.S}
Write 
$\theory S$ for $\theory{EQU}\cup\theory{SUB}\cup\theory{FOL}$.
\end{nttn}

\begin{lemm}
\label{lemm.amod.ssm}
$\begin{array}[t]{r@{\ }l}
\theory S\cent\amod{t'[a\ssm t]}\pnleq&\amod{t'}[a\sm \amod{t}]\quad \text{and}
\\
\theory S\cent\amod{\xi[a\ssm t]}\pnleq&\amod{\xi}[a\sm \amod{t}].
\end{array}
$
\end{lemm}
\begin{proof}
By routine inductions on $t$ and $\xi$.
\end{proof}

\begin{thrm}[Correctness]
\label{thrm.correctness}
If $\Xi\cent\Chi$ is derivable in first-order logic then
$\theory S\cent\amod{\Xi\cent \Chi}$ is derivable in PNL.
\end{thrm}
\begin{proof}
By a long but routine inspection we can check that \theory{EQU}, \theory{SUB}, and \theory{FOL} allow us to model the behaviour of `real' first-order logic.
We use Lemma~\ref{lemm.amod.ssm}.
\end{proof}

\subsection{Interpretation of first-order logic}

We recall the usual definition of interpretations in first-order logic:
\begin{frametxt}
\begin{defn}
A \deffont{(first-order logic) interpretation} $\mathcal M$ is
a \deffont{carrier set} $M$, and  
elements 
$$
0^\mden \in M, \quad
\f{succ}^\mden\in M\to M, \quad
+^\mden\in (M\times M)\to M, \quad\text{and}\quad
*^\mden\in (M\times M)\to M.
$$
\end{defn}
\end{frametxt}
It is convenient to fix some $\mathcal M$ from here until Theorem~\ref{thrm.arithmetic}.

\begin{defn}
\label{defn.valu.M}
Define $\f{Valu_{\mathbb A_\nu}}(M)$ by:
$$
\label{eq.valu.M}
\f{Valu_{\mathbb A_\nu}}(M)=\{\varepsilon \in \mathbb A_\nu\to M\mid \Exists{A\subseteq\mathbb A_\nu}A\ \text{finite}\ \land \ 
\Forall{a,b\not\in A}\varepsilon(a)=\varepsilon(b)\}
$$
Call elements of $\f{Valu_{\mathbb A_\nu}}(M)$ \deffont{$\mathbb A_\nu$-valuations} (to $M$).
$\varepsilon$ will range over $\mathbb A_\nu$-valuations.

If $x\in M$ write $\varepsilon[a\ssm x]$ for the valuation mapping $b$ to $\varepsilon(b)$ and mapping $a$ to $x$:
$$
\begin{array}{r@{\ }l}
\varepsilon[a\ssm x](a)=&x
\\
\varepsilon[a\ssm x](b)=&\varepsilon(b)
\end{array}
$$
Give $\varepsilon\in \f{Valu_{\mathbb A_\nu}}(M)$ and $X\subseteq\f{Valu_{\mathbb A_\nu}}(M)$ a \deffont{pointwise} permutation action: 
$$
\begin{array}{r@{\ }l}
(\pi\act\varepsilon)(a)=&\varepsilon(\pi\mone(a)) .
\\
\pi\act X =& \{\pi\act \varepsilon\mid \varepsilon\in X\} .
\end{array}
$$
$U,V$ will range over \emph{finitely-supported} subsets of $\f{Valu_{\mathbb A_\nu}}(M)$---so there exists some finite $A\subseteq\mathbb A_\nu$ such that for all $\pi$, if $\pi(a)=a$ for all $a\in A$ then $\pi\act U=U$.
\end{defn}

\begin{rmrk}
$\f{Valu_{\mathbb A_\nu}}(M)$ would normally just be called `the set of valuations'. 
We are more specific since we separately also have valuations on unknowns $X$ (Definition~\ref{defn.valuation}).

PNL atoms are serving as variable symbols of $\mathcal L$.  To
conveniently apply nominal techniques, it is useful to restrict to
valuations that are finite in the sense given in Definition~\ref{defn.valu.M}.
In any case, any term or formula will only contain finitely many
atoms.
\end{rmrk}

\begin{defn}
\label{defn.interpret.L} 
We extend the interpretation  to first-order logic syntax as follows:
\begin{frameqn}
\begin{array}{r@{\ }l}
\denot{\mathcal M}{\varepsilon}{a}
=& \varepsilon(a)
\\
\denot{\mathcal M}{\varepsilon}{0}
=&0^\mden
\\
\denot{\mathcal M}{\varepsilon}{\f{succ}(t)} =&\f{succ}^\mden(\denot{\mathcal M}{\varepsilon}{t})
\\
\denot{\mathcal M}{\varepsilon}{t'+t}=&+^\mden(\denot{\mathcal M}{\varepsilon}{t'},\denot{\mathcal M}{\varepsilon}{t})
\\
\denot{\mathcal M}{\varepsilon}{t'*t}=&
*^\mden(\denot{\mathcal M}{\varepsilon}{t'},\denot{\mathcal M}{\varepsilon}{t})
\\
\denot{\mathcal M}{\varepsilon}{\bot} =& 0
\\
\denot{\mathcal M}{\varepsilon}{\xi'\limp\xi}=& 
\f{max}\{1{-}\denot{\mathcal M}{\varepsilon}{\xi'},\denot{\mathcal M}{\varepsilon}{\xi}\} 
\\
\denot{\mathcal M}{\varepsilon}{\Forall{a}\xi}=&
\f{min}\{\denot{\mathcal M}{\varepsilon[a\ssm x]}{\xi}\mid x\in M\}
\\
\denot{\mathcal M}{\varepsilon}{t'\oeq t}=& 1\text{ if }
\denot{\mathcal M}{\varepsilon}{t'} = \denot{\mathcal M}{\varepsilon}{t}
\text{ and }0\text{ otherwise}
\end{array}
\end{frameqn}
\end{defn}

\begin{defn}
Call the formula $\xi$ \deffont{valid} in ${\mathcal M}$ when 
$\denot{\mathcal M}{\varepsilon}{\xi} = 1$ for all $\varepsilon$. 

Call $\xi_1, \ldots, \xi_k \vdash \chi_1, \ldots, \chi_l$ \deffont{valid} 
in ${\mathcal M}$ when 
$(\xi_1\land\ldots\land\xi_k)\limp(\chi_1\lor\ldots\lor\chi_l)$ is valid.
\end{defn}

\subsection{A theory of arithmetic in $\mathcal L$}

\begin{figure*}[t]
$$
\begin{array}[t]{lr@{\ }l}
\rulefont{ps0}
&\Forall{a}&\f{succ}(a)\oeq 0 \limp \bot 
\\
\rulefont{pss}
&\Forall{a',a}&\f{succ}(a)\oeq \f{succ}(a') \limp  a\oeq a'
\\
\rulefont{p{+}0}
&\Forall{a}&a+0\oeq a
\\
\rulefont{p{+}succ}
&\Forall{a',a}&a'+\f{succ}(a)\oeq \f{succ}(a')+a
\\
\rulefont{p{*}0}
&\Forall{a}&a*0\oeq 0
\\
\rulefont{p{*}succ}
&\Forall{a',a}&a'*\f{succ}(a)\oeq (a'*a)+a
\\[1ex]
\rulefont{pind}
&&\xi[a\ssm 0]\limp 
\\
&&\ \ \Forall{a}(\xi\limp \xi[a\ssm \f{succ}(a)]) \limp 
\\
&&\ \ \ \ \Forall{a}\xi\quad\qquad\qquad\qquad\qquad\qquad\text{(every $\xi$, every $a$)}
\end{array}
$$
\caption{\theory{arithmetic}: axioms for arithmetic in first-order logic}
\label{fig.fol.arithmetic}
\end{figure*}

\begin{frametxt}
\begin{defn}
\label{defn.model.arithmetic}
Define a first-order theory of \deffont{arithmetic} by the axioms in Figure~\ref{fig.fol.arithmetic}.

An interpretation is a \deffont{model} of arithmetic when 
$\denot{\mathcal M}{}{\xi} = 1$ 
for $\xi$ each of \rulefont{ps0}, \rulefont{pss}, \rulefont{p{+}0}, 
\rulefont{p{+}succ}, \rulefont{p{*}0}, \rulefont{p{*}succ}, 
and every instance of \rulefont{pind}.
\end{defn}
\end{frametxt}

\begin{rmrk}
\rulefont{pind} the induction axiom-scheme is of course of particular interest.
We therefore unpack what its validity 
$$
\label{eq.ind}
\denot{\mathcal M}{}{\xi[a\ssm 0]\limp \Forall{a}(\xi\limp
    \xi[a\ssm \f{succ}(a)]) \limp \Forall{a}\xi} = 1
\qquad\text{(every $\xi$, every $a$)}
$$
means, in a little more detail.
For every $a$ and $\xi$:
\begin{itemize*}
\item
If $\denot{\mathcal M}{\varepsilon}{\xi[a\ssm 0]} = 1$, and
\item
if for every $x\in M$,\ \ $\denot{\mathcal M}{\varepsilon[a\ssm x]}{\xi} = 1$
 implies that 
$\denot{\mathcal M}{\varepsilon[a\ssm x]}{\xi[a\ssm \f{succ}(a)]}
= 1$, 
\item
then for every $x\in M$,\ \ $\denot{\mathcal M}{\varepsilon[a\ssm x]}{\xi} = 1$.
\end{itemize*} 
In \rulefont{pind} we take `every $a$', and in \rulefont{PInd} we do not.
This is because in \rulefont{PInd}, $a$ is $\alpha$-convertible,  
\end{rmrk}

\subsection{Building an interpretation for $\dot{\mathcal L}$ out of one for $\mathcal L$}
\label{subsect.building}

Recall the PNL signature $\dot{\mathcal L}$ from Section~\ref{sect.pnl.arithmetic}.
Suppose $\mathcal M$ is a model of $\mathcal L$.
We use it to build an interpretation $\mathcal N$ of 
$\dot{\mathcal L}$.

\begin{defn}
Extend $\mathcal L$ to $\mathcal L{+}M$ where we add all elements of $M$ as constants, and extend the interpretation to interpret these constants as themselves in $M$.
(So if $x\in M$ then $x$ is a constant symbol in $\mathcal L{+}M$ and 
$\denot{\mathcal M}{\varepsilon}{x} = x$.)

Define an $\mathbb A_\nu$-valuation $\varepsilon_0\in\f{Valu_{\mathbb A_\nu}}(M)$ by 
$$
\varepsilon_0(a)=0^\mden\text{\quad always.}
$$ 

If $t$ is a term, we write $\denot{\mathcal M}{}{t}$ for the function 
$\lam{\varepsilon} \denot{\mathcal M}{\varepsilon}{t}$. 
If $\xi$ is a formula, we write $\denot{\mathcal M}{}{\xi}$ for the function 
$\lam{\varepsilon} \denot{\mathcal M}{\varepsilon}{\xi}$. 

We now define an interpretation $\mathcal N$ for $\dot{\mathcal L}$. 
We give a denotation to the base sorts $\iota$ and $o$ of $\dot{\mathcal L}$, as follows:
\begin{frameqn}
\begin{array}{r@{\ =\ }l@{\qquad}r@{\ =\ }l}
\iota^\nden & \{\denot{\mathcal M}{}{t} \mid t\text{ a term of }\mathcal L{+}M\}
&
\pi\act\denot{\mathcal M}{}{t}&\denot{\mathcal M}{}{\pi\act t} 
\\
o^\nden &
\{\denot{\mathcal M}{}{\xi} \mid \xi\text{ a formula of }\mathcal L{+}M\}
&
\pi\act\denot{\mathcal M}{}{\xi}&\denot{\mathcal M}{}{\pi\act \xi}
\end{array}
\end{frameqn}
\label{defn.sig.dot.L}
We give a denotation to the term-formers and proposition-formers of $\dot{\mathcal L}$, as follows:

\begin{frameqn}
\hspace{-1em}\begin{array}{r@{\,=\,}l}
{\tf{var}}^\nden\, a\, \varepsilon & \varepsilon(a) 
\\
\dotzero^\nden \,\varepsilon &0^\mden
\\
(\dotsucc^\nden \,u)\,\varepsilon &\f{succ}^\mden(u\varepsilon)
\\
\dotplus^\nden\,(u,v)\,\varepsilon &+^\mden(u\varepsilon,v\varepsilon)
\\
\dottimes^\nden\,(u,v)\,\varepsilon &*^\mden(u\varepsilon,v\varepsilon)
\\
\tf{sub}_\iota^\nden\,([a]u,v)\,\varepsilon & u(\varepsilon[a{\ssm} v\varepsilon])
\\
{\dotbot}^\nden \,\varepsilon & 0
\end{array}
\quad
\begin{array}{r@{\,=\,}l}
\tf{sub}_o^\nden\,([a]u,v)\,\varepsilon  & u(\varepsilon[a{\ssm} v\varepsilon])
\\
{\dotlimp}^\nden\,(U,V)\,\varepsilon  & 
\f{max}\{1{-}U\varepsilon,V\varepsilon\}
\\
{\dotforall}^\nden\,([a]U)\,\varepsilon  &  \f{min} \{U(\varepsilon[a{\ssm} x])\mid x \in M\}
\\
{\dot{\oeq}}^\nden\,(u,v)\,\varepsilon  & {\oeq^\mden}(u\varepsilon,v\varepsilon)
\\
\oeq_\iota^\nden\,(u,v) & 1\text{ if }u{=}v\text{ and }0\text{ otherwise}
\\
\oeq_o^\nden\,(U,V) & 1\text{ if }U{=}V\text{ and }0\text{ otherwise}
\\
\epsilon^\nden\,U & U(\varepsilon_0)
\end{array}
\end{frameqn}
Here, $u$ and $v$ range over $\iota^\nden$ and $U$ and $V$ range over $o^\nden$.
We insert brackets where this might increase clarity.  
\end{defn}

\maketab{tab3}{@{\hspace{-1em}}R{3em}@{\ }L{7em}@{\ }R{3em}@{\ }L{20em}}
\begin{rmrk}
For the reader's convenience we indicate the types of some of the symbols above in an informal `crib sheet':
\begin{tab3}
\varepsilon:&\mathbb A_\nu {\to} M
&
{\tf{var}}^\nden:&\mathbb A_\nu {\to} (\mathbb A_\nu{\to} M) {\to} M
\\
u:&(\mathbb A_\nu{\to} M){\to} M
&
U:&(\mathbb A_\nu{\to} M){\to} \{0,1\}
\\
0^\mden:&M
&
\dotzero^\nden:&(\mathbb A_\nu{\to} M){\to} M
\\
{\oeq^\mden}:&M^2{\to}\{0,1\}
&
{\dot{\oeq}}^\nden:&((\mathbb A_\nu{\to}M){\to}M)^2{\to}(\mathbb A_\nu{\to}M){\to}\{0,1\}
\\
&
&
{\dotforall}^\nden:&[\mathbb A_\nu]((\mathbb A_\nu{\to}M){\to}\{0,1\}) \to (\mathbb A_\nu{\to}M){\to}\{0,1\}
\end{tab3}
This crib sheet is only indicative since, for instance, $\varepsilon$ is not \emph{any} function in $\mathbb A_\nu{\to}M$ (see Definition~\ref{defn.valu.M}).
\end{rmrk}

\begin{lemm}
\label{lemm.compose}
\begin{enumerate}
\item
$\denot{\mathcal M}{\varepsilon}{t'[a\ssm t]} =
\denot{\mathcal M}
{\varepsilon[a\ssm \denot{\mathcal M}{\varepsilon}{t}]}{t'}
$. 
\item
$\denot{\mathcal M}{\varepsilon}{\xi[a\ssm t]} = 1$
\ if only if\   
$\denot{\mathcal M}
{\varepsilon[a\ssm \denot{\mathcal M}{\varepsilon}{t}]}
{\xi}
 = 1$. 
\end{enumerate}
\end{lemm}

\begin{lemm}
\label{lemm.some.equalities}
The following equalities all hold:
$$
\begin{array}[t]{r@{\ }l} 
\tf{var}^\nden(a)=&\denot{\mathcal M}{}{a}
\\
\dotzero^\nden=&\denot{\mathcal M}{}{0}
\\
\dotsucc^\nden(\denot{\mathcal M}{}{t})=&\denot{\mathcal M}{}{\f{succ}(t)}
\\
\dotplus^\nden(\denot{\mathcal M}{}{t'},\denot{\mathcal M}{}{t})=&\denot{\mathcal M}{}{t'+t}
\\
\dottimes^\nden(\denot{\mathcal M}{}{t'},\denot{\mathcal M}{}{t})=&\denot{\mathcal M}{}{t'*t}
\end{array}
\qquad
\begin{array}[t]{r@{\ }l} 
\tf{sub}_\iota^\nden([a]\denot{\mathcal M}{}{t'},\denot{\mathcal M}{}{t})=&\denot{\mathcal M}{}{t'[a\ssm t]}
\\
\tf{sub}_o^\nden([a]\denot{\mathcal M}{}{\xi},\denot{\mathcal M}{}{s})=& \denot{\mathcal M}{}{\xi[a\ssm s]}
\\
{\dotbot}^\nden=&\denot{\mathcal M}{}{\bot}
\\
{\dotlimp}^\nden(\denot{\mathcal M}{}{\xi'},\denot{\mathcal M}{}{\xi})=&\denot{\mathcal M}{}{\xi'\limp\xi}
\\
{\dotforall}^\nden([a]\denot{\mathcal M}{}{\xi})=&\denot{\mathcal M}{}{\Forall{a}\xi}
\\
{\dotoeq}^\nden(\denot{\mathcal M}{}{r},\denot{\mathcal M}{}{s})=&\denot{\mathcal M}{}{r\oeq s}
\end{array}
$$
\end{lemm}
\begin{proof}
We compare Definitions~\ref{defn.sig.dot.L} and~\ref{defn.interpret.L}.
Most cases are immediate; we consider only the slightly less trivial ones:
\maketab{tab1}{@{\hspace{-0em}}R{8em}@{\ \ $=$\ \ }L{12em}L{10em}}
\begin{tab1}
{\tf{var}}^\nden(a)&(\lam{a}\lam{\varepsilon}\varepsilon(a))a
&\text{Definition~\ref{defn.sig.dot.L}}
\\
&(\lam{a}\denot{\mathcal M}{}{a})a &\text{Definition~\ref{defn.interpret.L}}
\\
&\denot{\mathcal M}{}{a}&\text{fact} 
\\[1.5ex]
\tf{sub}_\iota^\nden([a]\denot{\mathcal M}{}{t'},\denot{\mathcal M}{}{t})
&\lam{\varepsilon}\denot{\mathcal M}{}{t'}(\varepsilon[a\ssm \denot{\mathcal M}{}{t}\varepsilon])
&\text{Definition~\ref{defn.sig.dot.L}}
\\
&\lam{\varepsilon}\denot{\mathcal M}{}{t'[a\ssm t]} &\text{Lemma~\ref{lemm.compose}}
\end{tab1}
Other cases are no harder.
\end{proof} 

\begin{lemm}
$\mathcal N$ (Definition~\ref{defn.sig.dot.L}) is a PNL interpretation.
\end{lemm}
\begin{proof}
We must check that:
\begin{itemize*}
\item
\emph{$\iota^\nden$ and $o^\nden$ are permissive-nominal sets.}

By routine calculations.
(In fact, $\iota^\nden$ and $o^\nden$ 
are \emph{nominal} sets; that is, their elements all have finite support.)
\item
\emph{The functions defined in Definition~\ref{defn.sig.dot.L} map elements of 
$\iota^\nden$,\ 
$o^\nden$,\ $[\mathbb A]\iota^\nden$,\ and $[\mathbb A]o^\nden$ correctly to the appropriate sets.}

By Lemma~\ref{lemm.some.equalities}.
\item
\emph{$\epsilon^\nden$ is equivariant from $o^\nden$ to $\{0,1\}$.}

By routine calculations using the fact that $(a\ b)\act\varepsilon_0=\varepsilon_0$.
\qedhere\end{itemize*}
\end{proof}

\begin{lemm}
\label{lemm.undot}
If $\amod{\Xi\cent\Chi}$ is valid in ${\mathcal N}$, then 
$\Xi\cent\Chi$ is valid in $\mathcal M$.
\end{lemm}
\begin{proof}
We calculate that if $\amod{\Xi\cent\Chi}$ is valid in $\mathcal N$, then
$$
\ \denot{\mathcal
  M}{\varepsilon_0}{(\xi_1\land\ldots\land\xi_k)\limp(\chi_1\lor\ldots\lor\chi_l)} = 1
$$ 
But the proposition written out above is closed, so for all valuations
$\varepsilon$, 
\ $
\denot{\mathcal M}{\varepsilon}{(\xi_1\land\ldots\land\xi_k)\limp(\chi_1\lor\ldots\lor\chi_l)}= 1 .
$ 
\end{proof}

Recall from Notation~\ref{nttn.S} that we write $\theory S$ for $\theory{EQU}\cup\theory{SUB}\cup\theory{FOL}$.
Recall also from Definition~\ref{defn.amod} the mapping $\amod{\text{-}}$ from first-order logic $\mathcal L$ to PNL terms. 
\begin{prop}
\label{prop.dot.M.arith}
$\mathcal N$ is a model of $\theory S\cup\theory{ARITH}$.
\end{prop}
\begin{proof}
By a routine verification.  We consider the axiom
\rulefont{\dotforall} from Figure~\ref{fig.fol}.  We unpack
definitions and see that we must prove that 
for every $\xi$ in
$\mathcal L{+}M$,\ 
\begin{itemize*}
\item
$\Forall{x\in M}\varepsilon_0[a\ssm x] \in \denot{\mathcal N}{}{\xi}$ 
if and only if 
\item
$\varepsilon_0[a\ssm \denot{\mathcal N}{\varepsilon_0}{t}]\in\denot{\mathcal N}{}{\xi}$ 
for every $t$ a term of $\mathcal L{+}M$.
\end{itemize*}
This follows, because $\mathcal L{+}M$ has a constant symbol for every
$x\in M$.  Validity of the other axioms is no harder.
\end{proof}

\begin{frametxt}
\begin{thrm} 
\label{thrm.arithmetic}
$\begin{array}[t]{l}\theory{arithmetic},\Xi\cent\Chi\text{ in first-order logic if and only if} 
\\
\theory S\cup\theory{ARITH}\cent \amod{\Xi\cent\Chi}\text{ in PNL}.
\end{array}
$ 
\end{thrm}
\end{frametxt}
\begin{proof}
We prove two implications.
The top-to-bottom implication follows using Theorem~\ref{thrm.correctness}.

For the bottom-to-top implication, we reason as follows:
Suppose 
$\theory S\cup\theory{ARITH}\cent \amod{\Xi\cent\Chi}$
in PNL.
Choose an arbitrary interpretation $\mathcal M$ of first-order logic that is a model of arithmetic, with carrier set $M$.
By Soundness (Theorem~\ref{thrm.soundness}) and Proposition~\ref{prop.dot.M.arith},\ $\amod{\Xi\cent\Chi}$ is valid in $\mathcal N$. 
By Lemma~\ref{lemm.undot} 
$\Xi\cent\Chi$ is valid in $\mathcal M$.
$\mathcal M$ was arbitrary, so by completeness of first-order logic \cite[\S 4.2]{shoen} it follows that $\Xi\cent\Chi$ is derivable. 
\end{proof}

\section{More PNL theories}
\label{sect.more.examples}

So far we have built PNL and used it to finitely axiomatise arithmetic.
In this section we briefly touch on how to express some known `nominal' constructs within PNL. 
 
\subsection{Inductive types}

Permissive-nominal logic can express the principles of nominal abstract syntax developed in \cite{gabbay:newaas-jv}.

Suppose a base sort $\iota$, a name sort $\nu$, and term-formers 
$$
\tf{var}:\nu\to \iota,\ \quad
\tf{app}:(\iota,\iota)\to \iota,\quad\text{and}\quad 
\tf{lam}:[\nu]\iota\to\iota.
$$
Fix an unknown $U:\iota$ and for brevity write $\phi[U\ssm r]$ as $\phi(r)$ for every $\phi$.
Suppose an axiom-scheme, for every $\phi$: 
$$
\begin{array}{l}
\phi(\tf{var}(a))
\limp
\\
\Forall{X}(\phi(X)\limp \phi(\tf{lam}([a]X))) 
\limp
\\
\Forall{X,Y}(\phi(X)\limp \phi(Y)\limp \tf{app}(X,Y))
\limp 
\\
\qquad\qquad\Forall{X}(\phi(X)) 
\end{array}
$$
Here $X$ and $Y$ have sort $\iota$ and we make a fixed but arbitrary choice of atom $a\in\pmss(X)$.

We can also express this finitely, if we axiomatise a sort for predicates (as we did for arithmetic).
Here is the axiom-scheme above made finite by using the theories \theory{EQU}, \theory{SUB}, and \theory{FOL} from Section~\ref{sect.pnl.arithmetic}:
$$
\begin{array}{r@{}l}
\Forall{Z}&\epsilon(Z[a\sm \tf{var}(a)])
\limp
\\
&\Forall{X}(\epsilon(Z[a\sm X])\limp \epsilon(Z[a\sm \tf{lam}([a]X))) 
\limp
\\
&\Forall{X,Y}(\epsilon(Z[a\sm X])\limp \epsilon(Z[a\sm Y])\limp \epsilon(Z[a\sm \tf{app}(X,Y)]))\limp
\\
&\qquad\qquad\Forall{X}\epsilon(Z[a\sm X])
\end{array}
$$

\subsection{The $\protect\new$ quantifier}
\label{subsect.new}

Nominal sets support the $\new$-quantifier \cite{gabbay:newaas-jv}.
PNL also includes the $\new$-quantifier; the way in which it does this is quite interesting, as we shall see in a moment.

$\new$ has some distinctive properties which are reflected in nominal logic (NL) and the logic of FM sets (FM):
\begin{align*}
&\begin{prooftree}
\Forall{x}(\tf P(x) \limp \New{a}\tf Q(a,x))
\Justifies
\Forall{x}\New{a}(\tf P(x)\limp \tf Q(a,x))
\end{prooftree}
&&
\begin{prooftree}
\Forall{x}\New{a}\New{b}(b\,a)\act x{\approx} x 
\Justifies
\New{a}\New{b}\Forall{x}(a\#x\limp b\#x\limp (b\,a)\act x{\approx} x) 
\end{prooftree}
\\
\intertext{Here and below we write a double horizontal line for `is provably equivalent to'.
$\new$ appears absent from Permissive-Nominal Logic (PNL).
It is `hiding' in the permission sets.
Corresponding propositions are, where $a,b\not\in\pmss(X)$:}
&
\begin{prooftree}
\Forall{X}(\tf P(X) \limp \tf Q(a,X))
\Justifies
\Forall{X}(\tf P(X)\limp \tf Q(a,X))
\end{prooftree}
&&
\begin{prooftree}
\Forall{X}(b\ a)\act X\approx X 
\Justifies
\Forall{X}(b\ a)\act X\approx X 
\end{prooftree}
\end{align*}
We see from these examples that two things are happening: first, freshness conditions are hard-coded into the syntax by permission sets---and second, so is the $\new$-quantifier.

It is interesting to consider another example.  In NL/FM:
\begin{alignat*}{4}
\qquad &
\begin{prooftree}
\New{a}\tf P(a)\land \New{a}\tf Q(b)
\Justifies
\New{a}\New{b}(\tf P(a)\land\tf Q(b))
\end{prooftree}
&\quad\quad\qquad&
\begin{prooftree}
\New{a}\tf P(a)\land \New{a}\tf Q(b)
\Justifies
\New{a}(\tf P(a)\land \tf Q(a))
\end{prooftree}
\qquad\qquad\qquad
\\
\intertext{Correspondingly in PNL we have:}
&
\begin{prooftree}
\tf P(a)\land \tf Q(b)
\Justifies
\tf P(a)\land\tf Q(b)
\end{prooftree}
&&
\begin{prooftree}
\tf P(a)\land \tf Q(b)
\Justifies
\tf P(a)\land \tf Q(a)
\end{prooftree}
\end{alignat*}
It is easy to use the rule \rulefont{Ax} from Figure~\ref{Seq} to construct a derivation proving that $\tf P(a)\land\tf Q(b)$ and $\tf P(a)\land \tf Q(a)$ are indeed logically equivalent in Permissive-Nominal Logic. 

The $\pi$ in \rulefont{Ax} expresses that truth is preserved by permutative renaming, or in symbols: $\cent \phi\liff \pi\act\phi$ always.

A permission set $S$ can be viewed in two ways: as giving permission to instantiate using free atoms in $S$---but also as a form of $\new$ for the atoms not in $S$.

\subsection{Freshness $a\#x$ and abstraction}
\label{subsection.freshness}

PNL has a notion of syntactic freshness which we identify as a notion of `free atoms of', and write $\f{fa}(t)$ and $\f{fa}(\phi)$.
Nominal sets also have a \emph{semantic} notion of freshness $a\#x$ given by $a\not\in\f{supp}(x)$.

As Lemma~\ref{lemm.supp.r} demonstrates, intuitively if $a$ is not free in $t$ then $a$ is fresh for the denotation of $t$.
In symbols: $a\not\in\f{fa}(t)\limp a\#\denot{}{}{t}$ (see \cite[Subsection~7.6]{gabbay:nomtnl} for a kind of converse). 

To capture in syntax the effect of a freshness predicate for a name sort $\nu$ on a sort $\alpha$, it suffices to assume \theory{EQU} (or extend PNL with an equality primitive) and to assume a predicate $\#$ of arity $(\nu,\alpha)$ with an axiom 
$$
\Forall{X}(a\#X\liff (b\ a)\act X\pnleq X) .
$$ 
Here $X$ has sort $\alpha$ and $a$ and $b$ have sort $\nu$, and $a\in\pmss(X)$ and $b\not\in\pmss(X)$.
This is essentially equation~13 in \cite{gabbay:newaas-jv}, using permission sets to attain the effect of the $\new$-quantifier. 
See also \cite[Subsection~5.2]{gabbay:nomuae} and \cite[Theorem~5.5]{gabbay:forcie} where similar observations were expressed for nominal algebra. 

Similarly atoms-abstraction $[a]r$ can be axiomatised not using atoms-abstraction as a term-former $\tf{abs}:(\nu,\tau)$ with axiom $\Forall{X}\tf{abs}(b,(b\ a)\act X)\pnleq \tf{abs}(a,X)$ where $b\not\in\pmss(X)$ (cf. \cite[Subsection~5.1]{gabbay:nomuae}). 

However, it is worthwhile to provide atoms-abstraction as primitive in PNL because it gives us access to PNL $\alpha$-equivalence which is a structural part of the PNL derivation system: to rename $a$ in $\tf{abs}(a,X)$ into $\tf{abs}(b,(b\ a)\act X)$ requires equality reasoning and axioms; to rename $a$ in $[a]X$ requires nothing but an $\alpha$-conversion (in this paper, they are actually the same term). 

\section{Cut-elimination}
\label{sect.cut-elimination}

Recall the \rulefont{Cut} rule from Figure~\ref{Seq}.
In this section we prove that \rulefont{Cut} is admissible in the presence of the other rules in Figure~\ref{Seq}.

\begin{defn}
\label{defn.Phi.sub}
Suppose $\f{fa}(r)\subseteq\pmss(X)$ and $r:\sort(X)$.
Define $\Phi[X\ssm r]$ by
$$
\Phi[X\ssm r]=\{\phi[X\ssm r]\mid \phi\in\Phi\}.
$$
\end{defn}

Lemmas~\ref{lemm.pi.theta} and~\ref{lemm.commute.sub} are proved by routine arguments like those in \cite{gabbay:perntu-jv,gabbay:nomu-jv}:
\begin{lemm}
\label{lemm.pi.theta}
$(\pi\act r)\theta\equiv \pi\act (r\theta)$.
\end{lemm}

\begin{lemm}
\label{lemm.commute.sub}
Suppose $Y\not\in\f{fV}(t)$.
Then 
$$
r[Y\ssm u][X\ssm t] \equiv r[X\ssm t][Y\ssm u[X\ssm t]].
$$
\end{lemm}

\begin{lemm}
\label{lemm.instantiate}
Suppose $\f{fa}(r)\subseteq\pmss(X)$ and $r:\sort(X)$.
Then 
$$
\Phi\cent\Psi\quad\text{implies}\quad
\Phi[X\ssm r]\cent \Psi[X\ssm r].
$$
\end{lemm}
\begin{proof}
By a routine induction on derivations.
The case of \rulefont{Ax} uses Lemmas~\ref{lemm.pi.theta} and~\ref{lemm.commute.sub}.
The case of \rulefont{\forall L} uses Lemma~\ref{lemm.commute.sub}.
\end{proof}

\begin{lemm}
\label{lemm.pi.right}
\begin{enumerate*}
\item
If there exists a derivation $\Deriv$ of $\Phi\cent\psi,\,\Psi$ then there exists a derivation of $\Phi\cent\pi\act\psi,\,\Psi$.
\item
If there exists a derivation $\Deriv$ of $\Phi,\,\phi\cent\Psi$ then there exists a derivation of $\Phi,\,\pi\act\phi\cent\Psi$.
\end{enumerate*}
\end{lemm}
\begin{proof}
By a simultaneous induction on $\Deriv$.
The case of \rulefont{\forall L} uses Lemma~\ref{lemm.pi.theta}.
(We need the \emph{simultaneous} induction for \rulefont{{\limp}L} and \rulefont{{\limp}R}, since parts of the proposition move between left and right.)
\end{proof}

\begin{nttn}
An instance of \rulefont{Cut} rests on two sub-derivations.
It is convenient to call them the \deffont{left branch} and \deffont{right branch} as illustrated:
$$
\begin{prooftree}
\begin{prooftree}
\leadsto
\Phi,\ \phi\cent \Psi
\using\text{\it Left branch}
\end{prooftree}
\qquad
\begin{prooftree}
\leadsto
\Phi\cent \phi, \Psi
\using\text{\it Right branch}
\end{prooftree}
\justifies
\Phi\cent\Psi
\using\rulefont{Cut}
\end{prooftree}
$$
\end{nttn}

\begin{thrm}[Cut-elimination]
\label{thrm.cut}
If $\Phi\cent\Psi$ is derivable with a derivation that uses \rulefont{Cut}, then it is derivable with a derivation that does not use \rulefont{Cut}. 
\end{thrm}
\begin{proof}
The proof is as for first-order logic.
The only differences are a $\pi$ in \rulefont{Ax} and a side-condition $\f{fa}(r)\subseteq\pmss(X)$ in \rulefont{\forall L}.
These affect terms and have no effect on the structure of derivations; for the purposes of this proof they are irrelevant.
 
We commute instances of \rulefont{Cut} upwards, as usual, following the method of \cite[pages 139-145]{elements} or \cite{gabbay:prottl}.
At each step, the following measure based on the depth of subderivations and the size of the cut formula, decreases:
\begin{itemize*}
\item
The size of the cut formula, and
\item
the longest path up the derivation the cut, that the formula persists, 
\end{itemize*}
lexicographically ordered.
\begin{itemize}
\item
The commutation cases between rules for $\limp$ and $\forall$ are as standard for first-order logic. 
\item
The essential case for $\limp$ is as standard.
\item
For the essential case for $\forall$, suppose the subderivation has the following form: 
$$
\begin{prooftree}
\begin{prooftree}
\Phi,\ \phi[X\ssm r]\cent \Psi
\justifies
\Phi,\ \Forall{X}\phi\cent\Psi
\using\rulefont{\forall L}
\end{prooftree}
\qquad
\begin{prooftree}
\begin{prooftree}
\leadsto
\Phi\cent \phi,\ \Psi
\using\Deriv
\end{prooftree}
\justifies
\Phi\cent \Forall{X}\phi,\ \Psi
\using\rulefont{\forall R}
\end{prooftree}
\justifies
\Phi\cent \Psi
\using\rulefont{Cut}
\end{prooftree}
$$
By Lemma~\ref{lemm.instantiate} there is a derivation $\Deriv[X\ssm r]$ of $\Phi\cent\phi[X\ssm r],\ \Psi$.
We eliminate the essential case as follows: 
$$
\hspace{-2.5em}
\begin{prooftree}
\Phi,\ \phi[X\ssm r]\cent \Psi
\qquad
\begin{prooftree} \leadsto \Phi\cent \phi[X\ssm r],\ \Psi \using \Deriv[X\ssm r] \end{prooftree}
\justifies
\Phi\cent \Psi
\using\rulefont{Cut}
\end{prooftree}
$$
\item
Suppose the subderivation has the following form: 
$$
\begin{prooftree}
\begin{prooftree}
\phantom{h}
\justifies
\Phi,\ \phi\cent \pi\act\phi,\Psi
\using\rulefont{Ax}
\end{prooftree}
\qquad
\begin{prooftree}
\leadsto
\Phi,\ \pi\act \phi \cent \Psi
\using\Deriv
\end{prooftree}
\justifies
\Phi,\ \phi\cent \Psi
\using\rulefont{Cut}
\end{prooftree}
$$
We use Lemma~\ref{lemm.pi.right} to obtain a derivation $\Deriv'$ of $\Phi,\ \phi\cent \Psi$ (the transformations involved in the proof of Lemma~\ref{lemm.pi.right} do not increase the inductive measure).
\end{itemize}
\end{proof}

\section{Related work}
\label{sect.related.work}

\subsection{Other `nominal' syntaxes and logics}
\label{subsect.using.nominal.sets}

Compared to other `nominal' logics, PNL emphasises ergonomics. 
When we use
PNL to axiomatize and build proofs in arithmetic or in set theory, we
want our hypotheses to speak about natural numbers or sets, not to speak
about atoms and freshness. 

Thus, although notions of atom and freshness are important, they should be implicit; i.e. handled automatically by the logic. 
This is the purpose of permission sets, which allow us to handle freshness and $\alpha$-renaming separately from logical deduction and equality reasoning.

\paragraph*{Axiomatisations}
The two best-known `nominal' logics are probably the \emph{nominal logic} of \cite{pitts:nomlfo-jv} and \emph{FM set theory}.
Both of these are Hilbert-style theories---sets of axioms---in first-order logic.
They are axiomatic theories of sets.

FM set theory contains axioms whose intended model is a sets cumulative hierarchy, whereas nominal logic contains axioms only for sets with a finitely-supported permutation action, with no assumption that they be composed of other sets.
For the purposes of this paper, the difference is not important.

\emph{Qua} logic, PNL is a logic whereas nominal logic and FM set theory are axiomatisations.
In addition and closely related to this, we can `just $\alpha$-rename' and `just choose a fresh atom'---as mentioned above we have  $\alpha$-renaming and freshness without appealing to equality reasoning and axioms.
 
\paragraph*{Proof-theories for the $\protect\new$-quantifier}
Natural deduction rules for $\new$ are proposed e.g. in \cite[Proposition~4.10]{gabbay:newaas-jv}, but these are not closed under substitution.
The second author created a proof-theory for $\new$ \cite{gabbay:frelog} with a good notion of proof-normalisation and a completeness proof, 
followed by an alternative treatment with Cheney \cite{gabbay:seqcnl}.

These gave $\new$ an operational behaviour as `pick a locally fresh name'; 
Cheney then developed another sequent system which gave $\new$ an operational behaviour as `pick a globally fresh name' \cite{cheney:simptn}.

The logic of \cite{cheney:simptn} includes
12 infinite axiom-schemes (Figures~3 and~4 of \cite{cheney:simptn}) 
describing the behaviour of atoms-abstraction from \cite{gabbay:newaas-jv}.
Thus, $\alpha$-equivalence (for atoms) is axiomatically handled and does not participate in the proof-theory.

PNL handles both $\alpha$-equivalence and the $\new$-quantifier very compactly, without recourse to axioms, and indeed requiring neither equality reasoning nor a $\new$-quantifier.

\paragraph*{One-and-a-halfth order logic} 
This logic is designed to represent schematic first-order reasoning (first-order derivations in the presence of `unknown predicates').
It corresponds roughly to the axiomatisation of first-order logic in Section~\ref{sect.pnl.arithmetic}.

\paragraph*{Semantic nominal terms}
In \cite{gabbay:semnt} we show how to interpret level 2 variables (unknowns) as infinite lists of distinct level 1 variables (atoms).
This allows us to build permissive-nominal term syntax as nominal abstract syntax-style inductive datatypes as proposed in \cite{gabbay:newaas-jv}.
The aim of this paper is to discuss the logic; not to analyse how its syntax could best be built.

\subsection{From nominal terms-in-freshness-context to PNL terms-with-permission-sets}
\label{subsect.tifc.pnl}

Nominal terms were introduced in \cite{gabbay:nomu-jv} where a decidable and efficient unification algorithm was demonstrated (see \cite{calves:polnua,levy:effnua,calves:phd} for the state of the art).
Nominal terms have been used in equational specification languages; in rewriting \cite{gabbay:nomr-jv} and in universal algebra (the logic of equality) \cite{gabbay:nomuae}.

PNL differs from nominal terms in three ways:
\begin{itemize*}
\item
Nominal terms use a finite \emph{freshness context} $a\#X$ whereas PNL uses permission sets, following \cite{gabbay:perntu,gabbay:perntu-jv} (a `permission set' $S$ can equally well be considered as a freshness sets $\mathbb A\setminus S$).
\item
PNL predicates include universal quantification over unknowns $\forall X$.
\item
PNL term syntax includes $\f{shift}$-permutations (the implications of this are discussed in Subsection~\ref{subsect.all.terms}).
\end{itemize*}
These features optimise PNL for being a first-order style logical foundation for mathematics with binding.

Another way to view this paper is as follows: PNL is the `obvious' extension of nominal algebra \cite{gabbay:nomuae} (an equational logic based on nominal terms), to a first-order logic.

But how then do we arrive at \emph{permissive} nominal techniques, starting from nominal algebra?
Suppose we have an equality axiom $a\#X\cent \tf f([a]X)=\tf g(X)$.
We want to write a corresponding first-order axiom.
There are two obvious routes to follow:
\begin{enumerate*}
\item
Assume first-order logic with a sort of atoms $\nu$ and some axioms (like nominal logic or FM set theory) and write 
$$
\Forall{z,x}z\#x\limp \tf f([z]x)=\tf g(x).
$$
Here $z$ and $x$ are variables and freshness $\#$ can be expressed using the axioms of the logic.

The problem with this is that we lose proof-theory; we are just working in a Hilbert-style axiom system.
\item
Imagine a first-order logic with nominal terms in which freshness conditions are attached to quantifiers, so that we can write 
$$
\Forall{X{:}\nu_{\#a}}\tf f([a]X)=\tf g(X). 
$$
Here $\nu_{\#a}$ means `elements of $\nu$ for which $a$ is fresh'.

The problem with this is that we have a poor theory of $\alpha$-equivalence; the freshness context for $X$ does not allow us to rename $[a]X$ to $[b](b\ a)\act X$, because there \emph{is} no $b$ fresh for $a$ (more on this in \emph{Implementing PNL} below).
\end{enumerate*}
Concerning the second option, we can add `freshening' axioms or derivation rules to the effect that $\Forall{X:\nu_{\#a}}\phi$ be logically equivalent to $\Forall{X{:}\nu_{\#a,b}}\phi$, and so on---this is in essence what the `freshening' rule \rulefont{fr} of nominal algebra (rule in Figure~2 of \cite{gabbay:nomuae}) does.
It should be possible to construct a version of PNL along these lines; the disadvantage would be that once $X$ is instantiated, we can no longer add further fresh atoms, unless we reintroduce \rulefont{fr} into PNL, but even then we would not recover full permissive-nominal $\alpha$-equivalence.

In PNL we take the idea to the $\omega$th degree; taking a limit of this `freshening' operation to obtain infinitely many fresh atoms, we arrive at permission sets.

\subsection{Non-nominal logics}

\paragraph*{First- and higher-order logic}
As discussed in the Introduction, we see PNL as sitting somewhere in-between these two logics.
It is more powerful---and we would claim more ergonomic---than first-order logic, because term-formers can bind.
Its advantage over higher-order logic is the smaller and simpler models and generally more first-order character.
Its term-syntax supports a decidable unification algorithm: both without $\f{shift}$ \cite{gabbay:perntu-jv} and with \cite{gabbay:nomtnl}.

\paragraph*{Logics based on the $\nabla$-quantifier}
A family of logics exists based on higher-order patterns and the $\nabla$-quantifier \cite{Miller:protgj,tiu:logrgj,gacek:comgjr}.
The intended meaning of e.g. $\nabla x.r{=}s$ is `$\lambda x.r{=}\lambda x.s$'.
Thus for instance the intended denotation of $\nabla x.\nabla y.x{=}y$ is $\lam{x}\lam{y}x{=}\lam{x}\lam{y}y$, and this is false.

As this example suggests, logics based on $\nabla$ use \emph{raising} and patterns (in brief: higher-order variables applied to finite lists of distinct variables, as in $xx_1\ldots x_n$) to obtain the effect of capturing substitution and variable dependencies, whereas we use permission sets and a two-level term syntax.
Our reading is that $\nabla$ is a way of peeling a single $\lambda$-abstraction uniformly off all terms and pushing it `into the meta-level'.
Or, to put it another way: $\nabla$ generates a fresh $\lambda$-abstracted variable.

The main philosophical difference here is that $\nabla$ is designed to assume $\alpha$-equivalence and treats variables as a `wire' which must always be bound, either by $\lambda$ in a term or possibly by a top-level $\nabla$; in contrast nominal techniques treat names as global and permutable and break $\alpha$-equivalence down into names and permutation.
In a separate journal paper submitted for publication, we relate these by translating permissive-nominal logic to higher-order logic in the style of e.g. a translation of permissive-nominal term unification to higher-order pattern unification \cite{gabbay:perntu-jv} or nominal algebra to algebra over higher-order terms \cite{gabbay:unialt}.

However, note that raising can cause a linear expansion in the size of a term (because what is represented by $X$ in this paper would be represented by $x x_1\ldots x_n$ in a logic based on raising), and can also cause `silly' $\beta$-redexes (since the mechanism which encodes dependency is the same mechanism which encodes computation).
This is one of the motivations for CMTT discussed below.

\paragraph*{Contextual modal type theory (CMTT)}
CMTT \cite{nanevski:conmtt} is a two-level system; typing contexts split into two halves; $\Delta$ and $\Phi$.
The two levels are different from the two levels used in (permissive-)nominal terms.
Variables $u:A[\Phi]\in\Delta$ range over representations of code; variables $x:A\in \Phi$ range over denotation.

This addresses a problem that is mostly orthogonal to what PNL tries to achieve. 
To the extent that this could be represented in PNL, it would be represented at the level of sorts---one sort for code, another for denotation (values).

Logics based on CMTT are consistent and have a well-studied proof-theory, so individual semantic models can be constructed using normal forms, but the question ``what is the general class of structures which these syntactic models represent?''' has no answer we know of.
That is, no general sets-based class of models has has been given for CMTT. 
Developing such models---perhaps using techniques borrowed from nominal sets---would be interesting future work. 

\paragraph*{Further logics}
Coming from other threads of research in computer science are logics designed to enrich first-order logic directly with binders without thinking specifically about inductive reasoning. 
We note in particular \emph{binding logic} \cite{dowek:binl} and \emph{$\lambda$-logic} \cite{beeson:laml}.

Binding logic enriched first-order terms with binders but forbade capture and turned out to be a little too weak.

$\lambda$-logic takes a direct approach of enriching first-order terms with $\lambda$-abstraction. 
The approach to binding taken by PNL is somewhat more general and is certainly different, in that it allows us to treat names as `bindable constants'.
That is, we can compare names for \emph{in}equality as names, while at the same time we can give them the behaviour of variables by axiomatising e.g. substitution for them, if we wish.

\section{Further remarks, further work}
\label{sect.further.remarks}

We have seen how permissive-nominal logic (PNL) extends first-order logic with term-formers that can bind.
We have given PNL a nominal-sets based semantics and shown it sound and complete.
We have considered a finite axiomatisation of arithmetic, based on a finite axiomatisation of first-order logic in PNL, and proven it correct.
Finally, we have proved cut-elimination.

In most respects PNL behaves just like first-order logic.
However, its `nominal' constructs let it ergonomically perform many of the tasks that would require much more powerful constructs in e.g. higher-order logic.

We do not claim that PNL makes nominal terms like those used in \cite{gabbay:nomu-jv,gabbay:nomr-jv} obsolete.
The argument for permissive-nominal terms is that they are a more abstract and powerful mathematical model with which to do proofs; for we use them to `just quotient' by $\alpha$-equivalence, and we use them to reconcile $\alpha$-equivalence of atoms with $\forall X$.
But that story is entirely compatible with prevous work.
For instance, though permission sets have the form $(\atomsdown \cup A)\setminus B$ (Definition~\ref{defn.atoms}), in practice we only seem to specify restrictions like $a\in \pmss(X)$ and $b\not\in \pmss(X)$---and these look like freshness environments.
In the discussion of Subsection~\ref{subsect.all.terms} we saw why: permission sets control capture, and we only care about controlling capture for the finitely many atoms mentioned explicitly in an axiom.
Furthermore the use of $\f{shift}$-permutations (Definition~\ref{defn.permutation}) means that the exact choice of permission set `does not really matter'---more on this below.
In a sense, freshness contexts live on in this paper and remain useful, as an emergent property of how we interact with a more abstract underlying mathematical model given by permissive-nominal terms.

We do not claim that PNL is the ultimate logic, whatever that means. 
However, we do see PNL as a significant step forward in the continuing search for logics suitable for axiomatising the systems with binding which are so common in computer science, as briefly discussed in Section~\ref{sect.related.work}.
We hope that PNL will turn out to be a `sweet spot' amongst such systems---fairly simple, yet usefully expressive and with good theoretic properties.

We will now briefly discuss some of the design decisions and design alternatives available to us in creating logics in the spirit of PNL.

\paragraph*{Unknowns of name sort, and atoms}

A swapping with unknowns, as in $(X\ Y)\act r$ where $X$ and $Y$ have a name sort $\nu$, is not primitive syntax in PNL.
This is \emph{atoms as variables} as in \cite{gabbay:seqcnl} and \cite{pitts:nomlfo-jv}, as opposed to the \emph{atoms as constants} approach of nominal terms \cite{gabbay:nomu-jv} which is inherited by this paper.

The axioms of nominal logic \cite[Section~5]{pitts:nomlfo-jv} can be copied over to endow term-formers $\tf{abs}$ and $\tf{swap}$ (see below) with the right properties; since the logic of \cite{pitts:nomlfo-jv} is already a set of axioms, there is no harm in doing this also in PNL.
Alternatively, we can `promote' behaviour from atoms to unknowns: 
Suppose a sort $\alpha$ and a name sort $\tf{Atm}$.
Suppose $\sort(X)=\alpha$ and $a\in\pmss(X)$.
Suppose a term-former $\tf{abs}:(\tf{Atm},\alpha)[\tf{Atm}]\alpha$.
The axiom 
$\Forall{X,Y}(X{\approx} a\limp \tf{abs}(X,Y){\approx} [a]Y)$
`promotes' atoms-abstraction, to an abstraction by unknowns of sort atom (over terms of sort $\alpha$).
Similarly for a term-former $\tf{swap}:(\tf{Atm},\tf{Atm},\alpha)\alpha$.

The question phrased in semantic terms is as follows: 
\begin{itemize*}
\item
Should every inhabitant of the denotation of a name sort $\nu$ to be referenced in PNL syntax by an atom?  
\item
Should every inhabitant of the semantics of name sorts, to be referenced by a closed term?
\end{itemize*}
The answers to these questions matter, but not for this paper.
It is not unusual (indeed, it is very common) for there to be more elements in a type than there are closed terms.
Furthermore, we have proved completeness in Subsection~\ref{subsect.completeness}, so any extra elements in name sorts cannot `make anything false'.

Yet it is reasonable to ask in future work whether we could exclude `non-standard atoms' in the same way that for example we might try to exclude non-standard numbers from models of first-order theories of arithmetic.
(Note that this is a general issue with a two-level syntax, and is not specific to PNL.)

We believe that this is possible. 
The idea is in the proof-theory for $\new$ from `Fresh Logic' \cite{gabbay:frelog}; see \rulefont{Exhaust\mathbb A} in Figure~3, and Subsection~5.5.

\paragraph*{Extending sorts}
It would be a good idea to introduce sort-formers and polymorphism into the sorting system, so that e.g. we can conveniently axiomatise a substitution action on an infinite class of sorts.
We see no difficulty in doing this---it is a definitional extension of what we already have.

Another interesting extension is to assume, for every sort $\alpha$, an associated name sort $\nu_\alpha$.
This would allow us to talk about `the level 1 variables associated with $\alpha$' in the same way that we can already talk about `the level 2 variables of sort $\alpha$'.

\paragraph*{Design of permission sets}

There is design freedom in the choice of permission sets.
We briefly sketch some of the options.

We can have \emph{more} permission sets. 
For instance, we can take $\mathbb A$ uncountable and permission sets all countably infinite sets.
We can also add finite permission sets, enabling us e.g. to reason about properties that only hold of (level 1) closed terms, or terms with a finitely bounded number of free atoms.
 
We can have \emph{fewer} permission sets.
For instance we could take permission sets to be $\atomsdown\setminus A$, or $\pi\act\atomsdown$ for finite $\pi$.

We can also have \emph{much} fewer permission sets---one, to be precise. 
PNL would work just as well if we took $\atomsdown$ as the single unique permission set.
The effect of larger or smaller permission sets can then be obtained using permutations.
For example if $\pmss(X)=\atomsdown$ and $\pmss(Y)=\f{shift}_\tau\act\pmss(X)$ and $\sort(X)=\tau=\sort(Y)$ then the effect of $\Forall{Y}\tf P(Y)$ can be obtained by the logically equivalent $\Forall{X}\tf P(\f{shift}_\tau\act X)$.
Using further $\f{shift}$-permutations and conjugation by finite permutations, any of the permission sets of Definition~\ref{defn.atoms} can be obtained.

Note that we never want $\mathbb A$ to be a permission set.
If we had that, then we would not be able to `choose a fresh atom' and e.g. would be unable to $\alpha$-convert $a$ in $[a]X$ if $\pmss(X)=\mathbb A$. 

Taking a more abstract view, a natural generalisation of Definition~\ref{defn.atoms} is an \emph{equivariant nominal join semi-lattice that does not contain $\top$}.
So specifically for sets of atoms, this means that if $S$ and $T$ are permission-sets then so are $\pi\act S$ and $S\cup T$, and $\mathbb A$ is not a permission set.
To illustrate how this works, note that if $S$ is a permission set and $\mathbb A\setminus S$ is finite then it easily follows using equivariance and sets unions that $\mathbb A$ is a permission set.  So to insist that $\mathbb A$ is not a permission set, is really to insist that every permission set is coinfinite.

The design decision made in Definition~\ref{defn.atoms} is simple, effective, and direct, and it allows us to express capture-avoidance conditions easily without complex `emulations' involving $\f{shift}$.

\paragraph*{PNL without $\f{shift}$}
We can restrict PNL by dropping $\f{shift}$-permutations (but retaining permission sets as defined), yielding
a logic that could be called \emph{PNL without shift}.
This is what was considered in the conference version of this paper \cite{gabbay:pernl}.

This is less ergonomic, but in a certain sense it is just as powerful.
It all depends on whether we want to be able to change our mind about a permission set in mid-derivation. 

This is a similar issue as appears e.g. in the design of a sequent system, whether we allow weakening as an explicit sequent rule (so that we can weaken mid-derivation), or integrate weakening into the axiom rule (so we have to anticipate the other propositions needed in the sequent).

The isomorphism between $\pmss(X)$ and $\pmss(X)\cup\{a\}$ for $a\not\in\pmss(X)$ is explicit in full PNL and an implicit fact in PNL without shift.

Unification of permissive-nominal terms without $\f{shift}$ was considered in \cite{gabbay:perntu-jv}.
Subsequently to writing this paper, that theory was re-cast using $\f{shift}$ \cite{gabbay:nomtnl}.

Note that nominal algebra satisfies an HSPA result whereas permissive-nominal algebra with shift satisfies an HSP result; details are elsewhere \cite{gabbay:nomahs,gabbay:nomtnl} but what is relevant to this discussion is that the extra expressivity which $\f{shift}$ gives, can make a real, mathematically measurable, difference.

\paragraph*{Implementing PNL}

An implementation of PNL could follow the lines of a first-order theorem-prover, since the proof-rules in Figure~\ref{Seq} are so like those of first-order logic.
The term-language would be richer and would include names and binding.

There would be many design choices, some of which we have touched on above: polymorphism in sorts; choice of permission sets (perhaps even adding variables to permission sets); whether or not to include $\f{shift}$; whether or not to exclude `non-standard' atoms using an \rulefont{Exhaust\mathbb A} rule like that in Figure~3 of \cite{gabbay:frelog}; and so on.

Another point is how much of the infinity of permission sets we should expose to the user.
To discuss this further, we must draw together several strands that have run through this paper from the beginning.

At the start of this paper we introduced permission sets, which guarantee infinite supplies of fresh atoms for every unknown. This culminated in Definition~\ref{defn.terms.and.propositions} with the permissive-nominal `just quotient' $\alpha$-equivalence, which is different from the notion of $\alpha$-equivalence used in nominal terms in e.g. \cite{gabbay:nomu-jv,gabbay:nomr-jv,gabbay:nomuae}. 
In Subsection~\ref{subsect.tifc.pnl} we gave a sense in which PNL is obtained from nominal terms and nominal algebra by adding universal quantifiers while taking a limit of extending freshness contexts in the sense of nominal terms.

But then at the start of this section we noted that for any \emph{concrete} derivation we only care about the finitely many atoms explicitly mentioned, thus for any concrete derivation we only care about finite freshness information after all.

So we have a choice, when implementing PNL, whether to \emph{present} the user with $\atomsdown$ and $\atomsup$ directly as we did in Definition~\ref{defn.atoms}, or to present a nominal terms syntax in which a (possibly but not necessarily finite) context of freshness assertions $\Delta$ is carried and may be extended as needed (nominal algebra does this using a freshness rule \rulefont{fr} \cite{gabbay:nomuae}; an idea taken from \cite{gabbay:frelog}).

A specific disadvantage is that we would lose the permissive-nominal `just quotient syntax' theory of $\alpha$-equivalence used in Definition~\ref{defn.terms.and.propositions}.

At the moment it seems unclear how much this matters from the point of view of an implementation.
After all, in an implementation we will have a specific goal with specific and (finitely many) atoms for which the user has chosen specific names. 
So will the user even appreciate explicit access to an infinite stock of fresh atoms?
Or, will the user prefer a freshness context to be extended as needed?
Note that the implementation might need fresh names when $\alpha$-renaming during resolution, so a resolution step might extend the freshness context with finitely but unboundedly many new names.
An advantage of presenting $\atomsup$ explicitly is that these names are honestly presented to the user from the start. 

As with any logic, there are many ways to present it.  Yet, the underlying mathematics remains essentially the same.

\paragraph*{Summary}

PNL addresses problems of mathematical specification with names and binding.
It provides a first-order logic environment which allows us to formally express the `informal meta-level', complete with names and binding.
As such, the most exciting potential application of PNL is as a logical foundation---as a meta-theory for mathematics---intermediate in power between first- and higher-order logic.
We believe that, perhaps with some fairly modest extensions, it would make an expressive, ergonomic, and practical alternative meta-language for mechanised mathematics.

\hyphenation{Mathe-ma-ti-sche}

\end{document}